\documentclass[aps,preprint]{revtex4}%
\usepackage{amssymb}
\usepackage{amsmath}
\usepackage{amsfonts}
\usepackage{graphicx}
\providecommand{\U}[1]{\protect\rule{.1in}{.1in}}

\newtheorem{theorem}{Theorem}

\newtheorem{definition}[theorem]{Definition}
\newenvironment{proof}[1][Proof]{\noindent\textbf{#1.} }{\ \rule{0.5em}{0.5em}}

\begin{document}

\title{Generating Higher-Order Lie Algebras by Expanding Maurer Cartan Forms.}
\author{Ricardo Caroca$^{1,2}$}
\email{rcaroca@ucsc.cl}
\author{Nelson Merino$^{1}$}
\email{nemerino@udec.cl}
\author{Alfredo P\'erez$^{1,3}$}
\email{perez@aei.mpg.de}
\author{Patricio Salgado$^{1}$}
\email{pasalgad@udec.cl}
\date{\today }

\begin{abstract}
By means of a generalization of the Maurer-Cartan expansion method we
construct a procedure to obtain expanded higher-order Lie algebras. The
expanded higher order Maurer-Cartan equations for the case
$\mathcal{G}=V_{0}\oplus V_{1}$ are found. 

A dual formulation for the S-expansion multialgebra procedure is also
considered. The expanded higher order Maurer Cartan equations are recovered
from S-expansion formalism by choosing a special semigroup. This dual
method could be useful in finding a generalization to the case of a
generalized free differential algebra, which may be relevant for physical
applications in, e.g., higher-spin gauge theories.
\end{abstract}

\affiliation{$^{1}$Departamento de F\'{\i}sica, Universidad de Concepci\'{o}n, Casilla 160-C, Concepci\'{o}n, Chile.\\
$^{2}$Departamento de Matem\'{a}tica y F\'{\i}sica Aplicadas, Universidad\\Cat\'{o}lica de la Sant\'{\i}sima Concepci\'{o}n, Alonso de Rivera 2850, Concepci\'{o}n, Chile. \\
$^{3}$ Max-Planck-Institut f\"{u}r Gravitationsphysik, Albert-Einstein-Institut, Am M\"{u}hlenberg 1, D-14476 Golm, Germany.
}

\maketitle

\section{Introduction}

A Lie algebra $\mathcal{G}$ with basis $\left\{  X_{i}\right\}  $, may be
realized by left-invariant generators $X_{i}$ on the corresponding group
manifold. If $C_{ij}^{\text{ \ }k}$ are the structure constants of
$\mathcal{G}$ in the basis $\left\{  X_{i}\right\}  ,$ then they satisfy
$\left[  X_{i},X_{j}\right]  =C_{ij}^{\text{ \ \ }k}X_{k}.$ If $\{\omega
^{i}(g)\}$, $i=1,...,r=dimG$, are the basis determined by the (dual,
left-invariant) Maurer--Cartan one-forms on $G$; then, the Maurer-Cartan
equations that characterize $\mathcal{G}$, in a way dual to its Lie bracket
description, are given by%
\begin{equation}
d\omega^{k}=-\frac{1}{2}C_{ij}^{\text{ \ }k}\omega^{i}\wedge\omega^{j},\text{
\ \ \ }i,j,k=1,...,r.\label{eins}%
\end{equation}

In direct analogy we can say that a higher-order Lie algebra $\left(
\mathcal{G},\left[  ,...,\right]  \right)  $ \cite{deazcarraga},
\cite{deazcarraga01}, \cite{deazcarraga0} with basis $\left\{  X_{i}\right\}
$, may be realized by left-invariant generators $X_{i}$ on the corresponding
group manifold. If $C_{i_{1}i_{2}\cdot\cdot\cdot\cdot i_{n}}^{\text{
\ \ \ \ \ \ \ \ \ \ \ \ \ }k}$ are the higher order structure constants of
$\left(  \mathcal{G},\left[  ,...,\right]  \right)  $ in the basis $\left\{
X_{i}\right\}  ,$ then they satisfy $\left[  X_{i_{1}},...,X_{i_{n}}\right]
=C_{i_{1}i_{2}\cdot\cdot\cdot\cdot i_{n}}^{\text{ \ \ \ \ \ \ \ \ \ \ \ \ \ }%
k}X_{k}$, where $\left[  X_{i_{1}},...,X_{i_{n}}\right]  $ are the called
higher order Lie bracket or multibracket$.$ If $\{\omega^{i}\}$,
$i=1,...,r=dimG$, are the basis determined by the (dual, left-invariant)
Maurer--Cartan one-forms on $G$; \ then, the generalized Maurer-Cartan
equations that characterize $\left(  \mathcal{G},\left[  ,...,\right]
\right)  $, in a way dual to its higher order Lie bracket description, are
given by
\begin{equation}
\tilde{d}\omega^{\sigma}=\frac{1}{\left(  2m-2\right)  !}\Omega_{i_{1}%
\cdot\cdot\cdot\cdot i_{2m-2}}^{\text{ \ \ \ \ \ \ \ \ \ \ \ \ \ \ \ \ }%
\sigma}\omega^{i_{1}}\wedge\cdot\cdot\cdot\wedge\omega^{i_{2m-2}},\label{zwei}%
\end{equation}
where $\widetilde{d_{m}}$ are the so-called higher-order exterior derivations.

As is noted in ref \cite{deazcarraga}, \cite{deazcarraga0} it could be
interesting to find applications of these higher-order Lie algebras to know
whether the cohomological restrictions which determine and conditions their
existence have a physical significance. \ Lie algebra cohomology arguments
have already been very useful in various physical problems as in the
description of anomalies or in the construction of the Wess-Zumino terms
required in the actions of extended supersymmetric objects. \ Other questions
may be posed from a purely mathematical point of view. From the discussion in
Sect.4 of ref. \cite{deazcarraga} we know that a representation of a simple
Lie algebra may not be a representation for the associated higher-order Lie
algebras. Thus, the representation theory of higher-order algebras requires a
separate analysis. A very interesting open problem from a structural point of
view is the expansions of higher-order Lie algebras, which will take us
outside the domain of the simple ones.

The purpose of this paper is to show that the expansion methods developed in
ref. \cite{deazcarraga1}, \cite{dualSexpansion} (see also \cite{hatsuda},
\cite{deazcarraga2}) can be generalized so that they permits obtaining new
higher-order Lie algebras of increasing dimensions from $\left(
\mathcal{G},\left[  ,...,\right]  \right)  $ by a geometric procedure based on
expanding the generalized Maurer Cartan equations.

The paper is organized as follows: In section 2 we shall review some aspects
of Generalized Maurer-Cartan equations. The main point of this section is to
display the differences between ordinary Maurer-Cartan equations and
Generalized Maurer-Cartan equations. In sections 3, 4 we generalize the
expansion methods developed in ref. \cite{deazcarraga1}, \cite{dualSexpansion}
and we give the general structure of the expansion method. \ Section 4 is
devoted to the dual S-expansion of higher-order Lie algebras. We close in
section 5 with conclusions and an outlook for future work.

\section{Generalized Maurer-Cartan equations}

In this section we shall review some aspects of the generalized Maurer-Cartan
equations. The main point of this section is to display the differences
between ordinary Maurer-Cartan equations and Generalized Maurer-Cartan
equations (see \cite{deazcarraga}, \cite{deazcarraga01}).
\begin{definition}
Let $\left\{  X_{i}\right\}  $ be a basis of $G$ given in terms of left
invariant vector fields on $G$, and $\wedge\ast(G)$ be the exterior algebra of
multivectors generated by them $\left(  X_{1}\wedge\cdot\cdot\cdot\wedge
X_{q}\equiv\varepsilon_{1\cdot\cdot\cdot\cdot q}^{i_{1}\cdot\cdot\cdot\cdot
i_{q}}X_{i_{1}}\otimes\cdot\cdot\cdot\otimes X_{i_{q}}\right)  $. The exterior
coderivation $\partial:\wedge^{q}\rightarrow\wedge^{q-1}$ is given by%
\begin{equation}
\partial\left(  X_{1}\wedge\cdot\cdot\cdot\wedge X_{q}\right)=\sum\limits_{\substack{l=1 \\l<k}}^{q}(-1)^{l+k+1}\left[  X_{l},X_{k}\right]
\wedge X_{1}\wedge\cdot\cdot\cdot\wedge\overset{\wedge}{X}_{l}\wedge\cdot
\cdot\cdot\wedge\overset{\wedge}{X}_{k}\cdot\cdot\cdot\wedge X_{q}.\label{uno}%
\end{equation}
\end{definition}

This definition is analogous to that of the exterior derivative $d$, as given
by the Palais formula \cite{deazcarraga01} with its first term missing when
one considers left-invariant forms (see $eq.(2.4)$ ref. \cite{deazcarraga01}).
As $d$, $\partial$ is nilpotent, $\partial^{2}=0$, due to the Jacobi Identity
for the commutator. In order to generalize (\ref{uno}), let us note that
$\partial(X_{1}\wedge X_{2})=[X_{1},X_{2}]$, so that (\ref{uno}) can be
interpreted as a formula that gives the action of $\partial$ on a $q$-vector
in terms of that on a bivector. For this reason we may write $\partial_{2}$
for $\partial$ above. It is then natural to introduce an operator
$\partial_{s}$ that on a $s$-vector gives the multicommutator of order $s$.
\begin{definition}
The general coderivation $\partial_{s}$ of degree $(s-1)$ ($s$ even)
$\partial_{s}$ : $\wedge^{n}(G)\rightarrow\wedge^{n-(s-1)}(G)$ is defined by
the action on an $n$-multivector%
\begin{align}
\partial_{s}\left(  X_{1}\wedge\cdot\cdot\cdot\wedge X_{n}\right)&=\frac{1}{s!}\frac{1}{\left(  n-s\right)  !}\varepsilon_{1\cdot\cdot\cdot\cdot
n}^{i_{1}\cdot\cdot\cdot\cdot i_{n}}\partial_{s}\left(  X_{i_{1}}\wedge
\cdot\cdot\cdot\wedge X_{i_{s}}\right)  \wedge X_{i_{s+1}}\wedge\cdot
\cdot\cdot\wedge X_{i_{n}},\label{drei} \\
\partial_{s}\wedge^{n}(G)&=0\text{ \ for }s>n, \\
\partial_{s}\left(  X_{1}\wedge\cdot\cdot\cdot\wedge X_{s}\right)  &=\left[
X_{1}\wedge\cdot\cdot\cdot\wedge X_{s}\right] , \\
\partial_{s}^{2}&\equiv0 .
\end{align}
\end{definition}%

We may now introduce the corresponding dual higher-order derivations
$\tilde{d}_{s}$ to provide a generalization of the Maurer-Cartan equations.
Since $\partial_{s}$ was defined on multivectors that are products of
left-invariant vector fields, the dual $\tilde{d}_{s}$ will be given for
left-invariant forms.

It is easy to introduce dual bases in $\wedge_{n}$ and in $\wedge^{n}$. With
$\omega^{i}(X_{j})=\delta_{ij}$, a pair of dual bases $\wedge_{n},$
$\wedge^{n}$are given by $\omega^{I_{1}}\wedge\cdot\cdot\cdot\wedge
\omega^{I_{n}},$ $\frac{1}{n!}X_{I_{1}}$ $\wedge...\wedge$ $X_{I_{n}}$
\ $(I_{1}<...<I_{n})$ since $\left(  \varepsilon_{j_{1}\cdot\cdot\cdot\cdot
j_{n}}^{i_{1}\cdot\cdot\cdot\cdot i_{n}}\omega^{j_{1}}\otimes\cdot\cdot
\cdot\otimes\omega^{j_{n}}\right)  \left(  \frac{1}{n!}\varepsilon_{l_{1}%
\cdot\cdot\cdot\cdot l_{n}}^{k_{1}\cdot\cdot\cdot\cdot k_{n}}X_{k_{1}}%
\otimes\cdot\cdot\cdot\otimes X_{k_{n}}\right)  =\varepsilon_{l_{1}\cdot
\cdot\cdot\cdot l_{n}}^{i_{1}\cdot\cdot\cdot\cdot i_{n}}$ y $\varepsilon
_{L_{1}\cdot\cdot\cdot\cdot L_{n}}^{I_{1}\cdot\cdot\cdot\cdot I_{n}}$ is $1$
if all indices coincide and $0$ otherwise. Nevertheless it is customary to use
the non-minimal set $\omega^{i_{1}}\wedge\cdot\cdot\cdot\wedge\omega^{i_{n}}$
to write \ $\alpha=\frac{1}{n!}\alpha_{i_{1}\cdot\cdot\cdot\cdot i_{n}}%
\omega^{i_{1}}\wedge\cdot\cdot\cdot\wedge\omega^{i_{n}}$. Since $\left(
\omega^{i_{1}}\wedge\cdot\cdot\cdot\wedge\omega^{i_{n}}\right)  \left(
X_{j_{1}},\cdot\cdot\cdot,X_{j_{n}}\right)  =\varepsilon_{j_{1}\cdot\cdot
\cdot\cdot j_{n}}^{i_{1}\cdot\cdot\cdot\cdot i_{n}},$ \ it is clear that
$\alpha_{i_{1}\cdot\cdot\cdot\cdot i_{n}}=\alpha\left(  X_{i_{1}},\cdot
\cdot\cdot,X_{i_{n}}\right)  =\frac{1}{n!}\left(  X_{i_{1}}\wedge\cdot
\cdot\cdot\wedge X_{i_{n}}\right)  .$
\begin{definition}
The action of $\tilde{d}_{m}$: $\wedge_{n}\rightarrow\wedge_{n+(2m-3)}$
(remember that $s=2m-2$) on $\alpha\in\wedge_{n}$is given by
\cite{deazcarraga01}
\begin{equation}
\tilde{d}_{m}\alpha\left(  X_{i_{1}},\cdot\cdot\cdot,X_{i_{n+2m-3}}\right)=\frac{1}{\left(  2m-2\right)  !\left(  n-1\right)  !}\varepsilon_{i_{1}%
\cdot\cdot\cdot\cdot i_{n+2m-3}}^{j_{1}\cdot\cdot\cdot\cdot j_{n+2m-3}}%
\alpha\left(  \left[  X_{j_{1}},\cdot\cdot\cdot,X_{j_{2m-2}}\right]
,X_{j_{2m-1}},\cdot\cdot\cdot X_{j_{n+2m-3}}\right)  ,\label{vier}%
\end{equation}%
\begin{equation}
\left(  \tilde{d}_{m}\alpha\right)  _{i_{1}\cdot\cdot\cdot\cdot i_{n+2m-3}%
}=\frac{1}{\left(  2m-2\right)  !\left(  n-1\right)  !}\varepsilon_{i_{1}%
\cdot\cdot\cdot\cdot i_{n+2m-3}}^{j_{1}\cdot\cdot\cdot\cdot j_{n+2m-3}}%
\Omega_{j_{1}\cdot\cdot\cdot\cdot j_{2m-2}}^{\text{
\ \ \ \ \ \ \ \ \ \ \ \ \ \ \ \ }\rho}\alpha_{\rho j_{2m-1}\cdot\cdot
\cdot\cdot j_{n+2m-3}}.\label{funf}%
\end{equation}
\end{definition}

From (\ref{vier}-\ref{funf}) we can see that the coordinates of
$\tilde{d}_{m}\omega^{\sigma}$ are given by%
\begin{align}
\left(  \tilde{d}_{m}\omega^{\sigma}\right)  \left(  X_{i_{1}},\cdot\cdot
\cdot,X_{i_{2m-2}}\right)  &=\frac{1}{\left(  2m-2\right)  !}\varepsilon
_{i_{1}\cdot\cdot\cdot\cdot i_{2m-2}}^{j_{1}\cdot\cdot\cdot\cdot j_{2m-2}%
}\omega^{\sigma}\left(  \left[  X_{j_{1}},\cdot\cdot\cdot,X_{j_{2m-2}}\right]
\right)  ,\label{sechs}   \\
\left(  \tilde{d}_{m}\omega^{\sigma}\right)  \left(  X_{i_{1}},\cdot\cdot
\cdot,X_{i_{2m-2}}\right)
&=\omega^{\sigma}\left(  \left[  X_{i_{1}},\cdot\cdot\cdot,X_{i_{2m-2}}\right]
\right)  =\omega^{\sigma}\Omega_{i_{1}\cdot\cdot\cdot\cdot i_{2m-2}}^{\text{
\ \ \ \ \ \ \ \ \ \ \ \ \ \ \ \ }\rho}X_{\rho}=\Omega_{i_{1}\cdot\cdot
\cdot\cdot i_{2m-2}}^{\text{ \ \ \ \ \ \ \ \ \ \ \ \ \ \ \ \ }\sigma
},\label{sieben}%
\end{align}
from which we conclude that%
\begin{equation}
\tilde{d}_{m}\omega^{\sigma}=\frac{1}{\left(  2m-2\right)  !}\Omega
_{i_{1}\cdot\cdot\cdot\cdot i_{2m-2}}^{\text{
\ \ \ \ \ \ \ \ \ \ \ \ \ \ \ \ }\sigma}\omega^{i_{1}}\wedge\cdot\cdot
\cdot\wedge\omega^{i_{2m-2}}.\label{dos}%
\end{equation}
For $m=2$, $\tilde{d}_{m}=-d$, equation (\ref{dos}) reproduces the usual
Maurer-Cartan equations. The equation (\ref{dos}) is called "the generalized
Maurer-Cartan equation". \ In the compact notation that uses the canonical
one-form $\theta$, we can say the the action of $\tilde{d}_{m}$ on the
canonical form $\theta$ is given by \cite{deazcarraga01}
\begin{equation}
\tilde{d}_{m}\theta=\frac{1}{\left(  2m-2\right)  !}\left[  \theta
,\theta,\overset{2m-2}{\cdot\cdot\cdot\cdot},\theta\right], \label{free index}%
\end{equation}
where the multibracket of form is defined by
\begin{equation}
\left[  \theta,\theta,\overset{2m-2}{\cdot\cdot\cdot\cdot},\theta\right]
=\omega^{i_{1}}\wedge\cdot\cdot\cdot\wedge\omega^{i_{2m-2}}\left[  X_{i_{1}%
},\cdot\cdot\cdot,X_{i_{2m-2}}\right]. \label{free index2}%
\end{equation}

Using Leibniz's rule for the $\tilde{d}_{m}$ operator we arrive at%
\begin{equation}
\tilde{d}_{m}^{2}\theta=-\frac{1}{\left(  2m-2\right)  !}\frac{1}{\left(
2m-3\right)  !}\left[  \theta,\overset{2m-3}{\cdot\cdot\cdot\cdot}%
,\theta,\left[  \theta,\overset{2m-2}{\cdot\cdot\cdot\cdot},\theta\right]
\right]  =0,
\end{equation}
which again expresses the Generalized Jacobi Identity.

\section{Expanding higher order Lie algebras by rescaling some coordinates of
the group manifold}

\subsection{\textbf{The higher-order Lie algebras} $\left(  \mathcal{G(N)}%
,\left[  ,...,\right]  \right)  $\textbf{\ generated from }$\left(
\mathcal{G},\left[  ,...,\right]  \right)  $, \textbf{when} $\mathcal{G}%
=V_{0}\oplus V_{1}$.}

The generalized Maurer-Cartan equations that characterize the multialgebra
$\left(  \mathcal{G},\left[  ,...,\right]  \right)  $, in a way dual to its
higher order Lie bracket description, are given by%
\begin{equation}
\tilde{d}_{m}\omega^{k}(g)=\frac{1}{\left(  2m-2\right)  !}C_{i_{1}%
...i_{2m-2}}^{k}\omega^{i_{1}}(g)\wedge...\wedge\omega^{i_{2m-2}%
}(g).\label{g1}%
\end{equation}

Consider the splitting of $\mathcal{G}^{\ast}$ into the sum of two vector
subspaces, $\mathcal{G}^{\ast}=V_{0}^{\ast}\oplus V_{1}^{\ast}$, \ where
$V_{0}^{\ast}$ and $V_{1}^{\ast}$ are generated by the Maurer-Cartan forms
$\omega^{i_{0}}\left(  g\right)  $ and $\omega^{i_{1}}\left(  g\right)  $ of
$\mathcal{G}^{\ast}$ with indices corresponding, respectively, to the
unmodified and modified parameters,%

\begin{equation}
g^{i_{0}}\rightarrow g^{i_{0}}\text{, \ \ }g^{i_{1}}\rightarrow\lambda
g^{i_{1}}\text{, \ \ }i_{0}\left(  i_{1}\right)  =1,...,\dim V_{0}\left(
V_{1}\right)  .\label{g2}%
\end{equation}

In general, the series of $\omega^{i_{0}}\left(  g,\lambda\right)  \in
V_{0}^{\ast}$ and $\omega^{i_{1}}\left(  g,\lambda\right)  \in V_{1}^{\ast},$
will involve all powers of $\lambda,$
\begin{equation}
\omega^{i_{p}}\left(  g,\lambda\right)  =\sum_{\alpha=0}^{\infty}%
\lambda^{\alpha}\omega^{i_{p},\alpha}\left(  g\right)  \text{, \ \ }%
p=0,1\text{.}\label{g3}%
\end{equation}

In terms of the 1-forms $\omega^{i_{p}}$, the generalized Maurer-Cartan
equations (\ref{g1}) take the form%

\begin{equation}
\tilde{d}_{m}\omega^{k_{s}}\left(  g\right)=\frac{1}{\left(  2m-2\right)  !}C_{\underset{\left(  2m-2\right)
\text{-indexes}}{\underbrace{i_{p}\ j_{q}\ m_{r}...n_{t}}}}^{k_{s}}%
\underset{\left(  2m-2\right)  \text{-indexes}}{\underbrace{\omega^{i_{p}%
}\left(  g\right)  \wedge\omega^{j_{q}}\left(  g\right)  \wedge\omega^{m_{r}%
}\left(  g\right)  ...\wedge\omega^{n_{t}}\left(  g\right)  }},\label{g4_0}%
\end{equation}
where, on the right side of equation (\ref{g1}) there is an implicit sum on
the (2m-2) indices $i_{p},j_{q},m_{r},\cdot\cdot\cdot,n_{t}=1,\cdot\cdot
\cdot,\dim V_{p}(V_{q})(V_{r})\cdot\cdot\cdot(V_{t})$ and on the (2m-2)
indices $p,q,r,\cdot\cdot\cdot,t=0,1$. Explicitly we have%

\[
\tilde{d}_{m}\omega^{k_{s}}\left(  g\right)  =\frac{1}{\left(  2m-2\right)
!}\sum_{p,q,r,...,t=0}^{1}\bullet
\]%
\begin{equation}
\bullet\sum_{i_{p},j_{q},m_{r},...,n_{t}=0}^{\dim V_{p}(V_{q})\left(
V_{r}\right)  ...(V_{t})}C_{\underset{\left(  2m-2\right)  \text{-indexes}%
}{\underbrace{i_{p}\ j_{q}\ m_{r}...n_{t}}}}^{k_{s}}\underset{\left(
2m-2\right)  \text{-indexes}}{\underbrace{\omega^{i_{p}}\left(  g\right)
\wedge\omega^{j_{q}}\left(  g\right)  \wedge\omega^{m_{r}}\left(  g\right)
...\wedge\omega^{n_{t}}\left(  g\right)  }}.\label{g4_01}%
\end{equation}

However, in general, we will consider the sums implicitly. We will denote the
set of (2m-2) indices \ $i,j,m,...,n=0,...,\dim\mathcal{G}$,
\ \ \ \ \ \ \ \ \ \ \ \ \ \ \ \ \ \ \ \ \ \ \ \ \ \ by $i^{l}$, where
$i^{1}=i,$ $i^{2}=j,$ $i^{3}=k,$ $\cdot\cdot\cdot,i^{2m-2}=n$ \ i.e.,
$i^{l}=0,...,\dim\mathcal{G}$; $l=1,...,2m-2.$ \ Since $\mathcal{G}%
=V_{0}\oplus V_{1},$ we have that the set of $(2m-2)$ indices
$p,q,r,...,t=0,1$ is useful to indicate that the forms $\omega^{i_{p}^{1}}$,
$\omega^{i_{q}^{2}}$,...,$\omega^{i_{t}^{2m-2}}$ belong to the subspaces
$V_{p}^{\ast}$, $V_{q}^{\ast}$,..., $V_{t}^{\ast}$ respectively. \ This allows
to denote the set of indices $p,q,r,...,t=0,1$ with the index $p_{l}$, where
the index $l$ reproduces the $(2m-2)$ indices: $p_{1}=p$, $p_{2}=q$, $p_{3}%
=r$,$\cdot\cdot\cdot\cdot,p_{2m-2}=t.$

With this notation, the generalized Maurer-Cartan equations \ (\ref{g4_0})
take the form%
\begin{equation}
\tilde{d}_{m}\omega^{k_{s}}\left(  g\right)  =\frac{1}{\left(  2m-2\right)
!}C_{i_{p_{1}}^{1}...i_{p_{2m-2}}^{2m-2}}^{k_{s}}\omega^{i_{p_{1}}^{1}}\left(
g\right)  \wedge...\wedge\omega^{i_{p_{2m-2}}^{2m-2}}\left(  g\right),
\label{g4}%
\end{equation}
where we have sumed over $i_{p_{l}}^{l}$ and over $p_{l}=0,1$ for every
$l=1,...,2m-2$. \ One might think that the super-index in $i_{p_{l}}^{l}$ is
superfluous. \ However the super-index $l$ is really necessary, for example,
to distinguish the independent sums existing over the indices $i_{p_{l}}^{l}$
and $i_{p_{l+1}}^{l+1}$ when $p_{l}=p_{l+1}$. In the compact notation
that uses the canonical one-form $\theta=\omega^{k_{s}}X_{k_{s}}$
\cite{deazcarraga01}, the eq. (\ref{g4}) can be written as
\begin{equation}
\tilde{d}_{m}\theta=\frac{1}{\left(  2m-2\right)  !}\left[  \theta
,\theta,\overset{2m-2}{\cdot\cdot\cdot\cdot},\theta\right], \label{corr1}%
\end{equation}
where the multibracket of forms is defined by
\begin{equation}
\left[  \theta,\theta,\overset{2m-2}{\cdot\cdot\cdot\cdot},\theta\right]
=\omega^{i_{p_{1}}^{1}}\wedge\cdot\cdot\cdot\wedge\omega^{i_{p_{2m-2}}^{2m-2}%
}\left[  X_{i_{p_{1}}^{1}},\cdot\cdot\cdot,X_{i_{p_{2m-2}}^{2m-2}}\right].
\label{corr3}%
\end{equation}

Following the procedure of Ref. \cite{deazcarraga1} we now insert the
expansions (\ref{g3}) into the Maurer-Cartan equations (\ref{g1}). After
tedious but direct calculation we obtain%
\begin{equation}
\tilde{d}_{m}\omega^{k_{s},\alpha}=\frac{1}{\left(  2m-2\right)  !}C_{\left(
i_{p_{1}}^{1},\beta^{1}\right)  \cdot\cdot\cdot\left(  i_{p_{2m-2}}%
^{2m-2},\beta^{2m-2}\right)  }^{\left(  k_{s},\alpha\right)  }\omega
^{i_{p_{1}}^{1},\beta^{1}}\wedge\cdot\cdot\cdot\wedge\omega^{i_{p_{2m-2}%
}^{2m-2},\beta^{2m-2}},\label{g8}%
\end{equation}
where on the right side, besides the sums over $i_{p_{l}}^{l}$ and $p_{l}$, a
sum exists over $\beta^{l}=0,1,\cdot\cdot\cdot\cdot,\alpha$ for every
$l=0,...,2m-2$ and where
\begin{equation}
C_{\left(  i_{p_{1}}^{1},\beta^{1}\right)  \cdot\cdot\cdot\cdot\cdot\left(
i_{p_{2m-2}}^{2m-2},\beta^{2m-2}\right)  }^{\left(  k_{s},\alpha\right)
}=C_{i_{p_{1}}^{1}\cdot\cdot\cdot\cdot i_{p_{2m-2}}^{2m-2}}^{k_{s}}%
\delta_{\beta^{1}+\cdot\cdot\cdot\cdot+\beta^{2m-2}}^{\alpha}\text{.}%
\label{g9}%
\end{equation}

In the compact notation that now uses the canonical one-form $\theta^{\left(  N\right)  }=\omega^{k_{s},\alpha}X_{k_{s},\alpha}%
$ of the expanded multialgebra, $\left(  \mathcal{G}\left(
N\right)  ,\left[  ,...,\right]  \right)  $, the eq. (\ref{g8}) can
be written as 
\begin{equation}
\tilde{d}_{m}\theta^{\left(  N\right)  }=\frac{1}{\left(  2m-2\right)
!}\left[  \theta^{\left(  N\right)  },\theta^{\left(  N\right)  }%
,\overset{2m-2}{\cdot\cdot\cdot\cdot},\theta^{\left(  N\right)  }\right],
\label{corr4}%
\end{equation}
where the multibracket of forms is defined by
\begin{equation}
\left[  \theta^{\left(  N\right)  },\theta^{\left(  N\right)  },\overset
{2m-2}{\cdot\cdot\cdot\cdot},\theta^{\left(  N\right)  }\right]
=\omega^{i_{p_{1}}^{1},\beta^{1}}\wedge\cdot\cdot\cdot\wedge\omega
^{i_{p_{2m-2}}^{2m-2},\beta^{2m-2}}\left[  X_{i_{p_{1}}^{1},\beta^{1}}%
,\cdot\cdot\cdot,X_{i_{p_{1}}^{1},\beta^{1}}\right]  .\label{corr5}%
\end{equation}
Note that $\left\{  X_{k_{s},\alpha}\right\}  $ is the basis of $\left(
\mathcal{G}\left(  N\right)  ,\left[  ,...,\right]  \right)  $ while $\left\{
\omega^{k_{s},\alpha}\right\}  $ is the dual basis.

The generalized Jacobi identity is obtained the calculation of $\tilde{d}%
_{m}^{2}\omega^{k_{s},\alpha}$. The equations (\ref{g8})-(\ref{g9}) are the
direct generalization to the case of higher order Lie algebras of the
equations $(2.15)$ of the Ref. \cite{deazcarraga1}.

The following theorem generalizes the theorem $1$ of Ref. \cite{deazcarraga1}
to the case of higher order Lie algebras and it establishes the conditions
under which the $1$-forms $\omega^{i_{0},\alpha_{0}},$ $\omega^{i_{1}%
,\alpha_{1}}$ generate new higher order Lie algebras.

\begin{theorem}
Let $\left(  \mathcal{G},\left[  ,...,\right]  \right)  $ be a higher order
Lie algebra and $\mathcal{G}=V_{0}\oplus V_{1}$ (no higher order Lie
sub-algebra conditions are assumed, neither for $V_{0}$ nor for $V_{1}$). Let
$\left\{  \omega^{i}\right\}  $, $\left\{  \omega^{i_{0}}\right\}  $,
$\left\{  \omega^{i_{1}}\right\}  $ ($i=1,\cdot\cdot\cdot,\dim\mathcal{G}$,
$i_{0}=1,\cdot\cdot\cdot\cdot,\dim V_{0}$, $i_{1}=1,\cdot\cdot\cdot\cdot,\dim
V_{1}$) be, respectively, the bases of the $\mathcal{G}^{\ast}$, $V_{0}^{\ast
}$ and $V_{1}^{\ast}$ dual vector spaces. Then, the vector space generated by
\begin{equation}
\left\{  \omega^{i_{0},0},\omega^{i_{0},1},\cdot\cdot\cdot,\omega^{i_{0}%
,N};\omega^{i_{1},0},\omega^{i_{1},1},\cdot\cdot\cdot\cdot,\omega^{i_{1}%
,N}\right\}, \label{g10}%
\end{equation}
together with the generalized Maurer-Cartan equations (\ref{g8}) for the
structure constants (\ref{g9}) determine a higher order Lie algebra
$\mathcal{G}\left(  N\right)  $ for each expansion order $N\geq0\,$\ of
dimension $\dim\mathcal{G}\left(  N\right)  =\left(  N+1\right)
\dim\mathcal{G}$.
\end{theorem}

\begin{proof}
The generalized Maurer-Cartan equations (\ref{g8}-\ref{g9}) can be written as
\begin{equation}
\tilde{d}_{m}\omega^{k_{s},\alpha}=\frac{1}{\left(  2m-2\right)  !}%
\sum\limits_{\beta^{1},\cdot\cdot\cdot\cdot,\beta^{2m-2}=0}^{\alpha
}C_{i_{p_{1}}^{1}\cdot\cdot\cdot\cdot i_{p_{2m-2}}^{2m-2}}^{k_{s}}%
\delta_{\beta^{1}+\cdot\cdot\cdot\cdot+\beta^{2m-2}}^{\alpha}\omega^{i_{p_{1}%
}^{1},\beta^{1}}\wedge\cdot\cdot\cdot\wedge\omega^{i_{p_{2m-2}}^{2m-2}%
,\beta^{2m-2}}.
\end{equation}
Let's remember that we have sums over $p_{l}$ and over $\beta^{l}$ such that
$\alpha=\beta^{1}+\cdot\cdot\cdot+\beta^{2m-2}.$ \ We can see that for
$\alpha=N_{0}$%
\begin{equation}
\tilde{d}_{m}\omega^{k_{0},N_{0}}=\frac{1}{\left(  2m-2\right)  !}%
\sum\limits_{\beta^{1},\cdot\cdot\cdot,\beta^{2m-2}=0}^{N_{0}}C_{i_{p_{1}}%
^{1}\cdot\cdot\cdot i_{p_{2m-2}}^{2m-2}}^{k_{0}}\delta_{\beta^{1}+\cdot
\cdot\cdot+\beta^{2m-2}}^{N_{0}}\omega^{i_{p_{1}}^{1},\beta^{1}}\wedge
\cdot\cdot\cdot\wedge\omega^{i_{p_{2m-2}}^{2m-2},\beta^{2m-2}},%
\end{equation}
appear in the sum terms that contain $1$-forms $\omega^{i_{_{1}}^{l},N_{0}}$,
whereas
\begin{equation}
\tilde{d}_{m}\omega^{k_{1},N_{1}}=\frac{1}{\left(  2m-2\right)  !}%
\sum\limits_{\beta^{1},\cdot\cdot\cdot,\beta^{2m-2}=0}^{N_{1}}C_{i_{p_{1}}%
^{1}\cdot\cdot\cdot i_{p_{2m-2}}^{2m-2}}^{k_{1}}\delta_{\beta^{1}+\cdot
\cdot\cdot+\beta^{2m-2}}^{N_{0}}\omega^{i_{p_{1}}^{1},\beta^{1}}\wedge
\cdot\cdot\cdot\wedge\omega^{i_{p_{2m-2}}^{2m-2},\beta^{2m-2}},%
\end{equation}
appear in the sum terms that contain $1$-forms $\omega^{i_{_{1}}^{l},N_{1}}. $
\ Wherefrom we see that the forms $\omega^{i_{_{1}}^{l},N_{0}}$ and
$\omega^{i_{_{1}}^{l},N_{1}}$, for any $l=1,\cdot\cdot\cdot,2m-2$, are in the
base (\ref{g10}),\ if and only if \ $N_{0}=N_{1}=N.$ This means that the set
\begin{equation}
\left\{  \omega^{i_{0},0},\omega^{i_{0},1},\cdot\cdot\cdot,\omega^{i_{0}%
,N};\omega^{i_{1},0},\omega^{i_{1},1},\cdot\cdot\cdot\cdot,\omega^{i_{1}%
,N}\right\},
\end{equation}
generates a higher order Lie algebra of dimension
\begin{equation}
\dim\mathcal{G}\left(  N\right)  =\left(  N+1\right)  \dim V_{0}+\left(
N+1\right)  \dim V_{1}=\left(  N+1\right)  \dim\mathcal{G}\text{.}%
\end{equation}

To prove that the generalized Jacobi identity is satisfied we calculate
$\tilde{d}_{m}^{2}\omega^{k_{s},\alpha}:$%
\begin{align}
\tilde{d}_{m}^{2}\omega^{k_{s},\alpha}&=\frac{1}{\left(  2m-2\right)
!}C_{\left(  i_{p_{1}}^{1},\beta^{1}\right)  \cdot\cdot\cdot\cdot\cdot\left(
i_{p_{2m-2}}^{2m-2},\beta^{2m-2}\right)  }^{\left(  k_{s},\alpha\right)
}\tilde{d}_{m}\left(  \omega^{i_{p_{1}}^{1},\beta^{1}}\wedge\cdot\cdot
\cdot\cdot\wedge\omega^{i_{p_{2m-2}}^{2m-2},\beta^{2m-2}}\right) , \nonumber \\
&=\frac{\left(  2m-2\right)  !}{\left(
2m-2\right)  !}C_{\left(  i_{p_{1}}^{1},\beta^{1}\right)  \cdot\cdot\cdot
\cdot\cdot\left(  i_{p_{2m-2}}^{2m-2},\beta^{2m-2}\right)  }^{\left(
k_{s},\alpha\right)  }\tilde{d}_{m}\omega^{i_{p_{1}}^{1},\beta^{1}}\wedge
\cdot\cdot\cdot\cdot\wedge\omega^{i_{p_{2m-2}}^{2m-2},\beta^{2m-2}}, \nonumber \\ 
&=\frac{1}{\left(  2m-2\right)  !\left(
2m-3\right)  !}C_{\left(  i_{p_{1}}^{1},\beta^{1}\right)  [\left(  i_{p_{2}%
}^{2},\beta^{2}\right)  \cdot\cdot\cdot\cdot\cdot\left(  i_{p_{2m-2}}%
^{2m-2},\beta^{2m-2}\right)  }^{\left(  k_{s},\alpha\right)  }C_{\left(
j_{p_{1}}^{1},\gamma^{1}\right)  \cdot\cdot\cdot\cdot\cdot\left(  j_{p_{2m-2}%
}^{2m-2},\gamma^{2m-2}\right)  ]}^{\left(  i_{p_{1}}^{1},\beta^{1}\right)  } \times \nonumber \\%
&\times\left(  \omega^{j_{p_{1}}^{1},\gamma^{1}}\wedge\cdot\cdot\cdot
\cdot\wedge\omega^{j_{p_{2m-2}}^{2m-2},\gamma^{2m-2}}\right)  \wedge\left(
\omega^{i_{p_{2}}^{2},\beta^{2}}\wedge\cdot\cdot\cdot\cdot\wedge
\omega^{i_{p_{2m-2}}^{2m-2},\beta^{2m-2}}\right)  =0.
\end{align}
Therefore
\begin{equation}
C_{\left(  i_{p_{1}}^{1},\beta^{1}\right)  [\left(  i_{p_{2}}^{2},\beta
^{2}\right)  \cdot\cdot\cdot\cdot\cdot\left(  i_{p_{2m-2}}^{2m-2},\beta
^{2m-2}\right)  }^{\left(  k_{s},\alpha\right)  }C_{\left(  j_{p_{1}}%
^{1},\gamma^{1}\right)  \cdot\cdot\cdot\cdot\cdot\left(  j_{p_{2m-2}}%
^{2m-2},\gamma^{2m-2}\right)  ]}^{\left(  i_{p_{1}}^{1},\beta^{1}\right)
}=0.\label{g11}%
\end{equation}
Introducing (\ref{g11}) into (\ref{g9}) we find
\begin{equation}
\delta_{\beta^{1}+\cdot\cdot\cdot\cdot+\beta^{2m-2}}^{\alpha}\delta
_{\gamma^{1}+\cdot\cdot\cdot\cdot+\gamma^{2m-2}}^{\beta^{1}}C_{i_{p_{1}}%
^{1}[i_{p_{2}}^{2}\cdot\cdot\cdot\cdot i_{p_{2m-2}}^{2m-2}}^{k_{s}}%
C_{j_{p_{1}}^{1}\cdot\cdot\cdot\cdot j_{p_{2m-2}}^{2m-2}]}^{i_{p_{1}}^{1}}=0,
\end{equation}
which is satisfied identically due to the validity of the generalized Jacobi
identity for the original multialgebra $\left(  \mathcal{G},\left[
,...,\right]  \right)  $.
\end{proof}

\subsection{\textbf{The case in which }$V_{0}$\textbf{\ is a subalgebra of
}$\mathcal{G}$\textbf{\ and of a submultialgebra }$\left(  \mathcal{G},\left[
,...,\right]  \right)  $}

Let $\left(  \mathcal{G},\left[  ,\right]  \right)  $ be a Lie algebra
and$\ $let$\ \ \left(  \mathcal{G},\left[  ,...,\right]  \right)  $ be a
higher order Lie algebra. We will assume that the vector space $\mathcal{G}%
=V_{0}\oplus V_{1}$ is such that $V_{0}$ is a subalgebra of $\left(
\mathcal{G},\left[  ,\right]  \right)  $ and a submultialgebra of $\left(
\mathcal{G},\left[  ,...,\right]  \right)  $. \ From Ref. \cite{deazcarraga1}
it is known that, if $V_{0}$ is a subalgebra, then%
\begin{align}
\omega^{i_{p}}\left(  g,\lambda\right)   &  =\sum_{\alpha=0}^{\infty}%
\lambda^{\alpha}\omega^{i_{p},\alpha}\left(  g\right), \label{t3}\\
\omega^{i_{p},\alpha} &  =0\text{, \ \ for }\alpha<p\text{.}
\end{align}
Introducing (\ref{t3}) into the generalized Maurer-Cartan equations we find
that, when $V_{0}$ is a subalgebra, the generalized expanded Maurer-Cartan
equations are given by%
\begin{equation}
\tilde{d}_{m}\omega^{k_{s},\alpha_{s}}=\frac{1}{\left(  2m-2\right)
!}C_{\left(  i_{p_{1}}^{1},\beta_{p_{1}}^{1}\right)  ...\left(  i_{p_{2m-2}%
}^{2m-2},\beta_{p_{2m-2}}^{2m-2}\right)  }^{\left(  k_{s},\alpha_{s}\right)
}\omega^{i_{p_{1}}^{1},\beta_{p_{1}}^{1}}\wedge...\wedge\omega^{i_{p_{2m-2}%
}^{2m-2},\beta_{p_{2m-2}}^{2m-2}},\label{t5}%
\end{equation}
where
\begin{equation}
C_{\left(  i_{p_{1}}^{1},\beta_{p_{1}}^{1}\right)  ...\left(  i_{p_{2m-2}%
}^{2m-2},\beta_{p_{2m-2}}^{2m-2}\right)  }^{\left(  k_{s},\alpha_{s}\right)
}=C_{i_{p_{1}}^{1}...i_{p_{2m-2}}^{2m-2}}^{k_{s}}\delta_{\beta_{p_{1}}%
^{1}+...+\beta_{p_{2m-2}}^{2m-2}}^{\alpha_{s}}\text{, \ \ \ }\label{t6}%
\end{equation}%
\begin{equation}
\alpha_{s}=0,\cdot\cdot\cdot,N_{s}\text{; }\beta_{p_{l}}^{l}=0,\cdot\cdot
\cdot,\text{ \ }p_{l}=0,1;\label{t6_2}%
\end{equation}%
\begin{equation}
\text{ }\omega^{i_{p_{l}}^{l},\beta_{p_{l}}^{l}}=0\text{, for }\beta_{p_{l}%
}^{l}<p_{l}\text{; }l=1,...,2m-2.\text{\ \ }\label{t7}%
\end{equation}
The equations (\ref{t5})-(\ref{t7}) are a direct generalization to a higher
order Lie algebra case of equations $(3.13)$-$(3.14)$ of Ref.
\cite{deazcarraga1}. In the compact notation that uses the canonical
one-form $\theta^{\left(  N\right)  },$ the eq. (\ref{t5}) can be
written as
\begin{equation}
\tilde{d}_{m}\theta^{\left(  N\right)  }=\frac{1}{\left(  2m-2\right)
!}\left[  \theta^{\left(  N\right)  },\theta^{\left(  N\right)  }%
,\overset{2m-2}{\cdot\cdot\cdot\cdot},\theta^{\left(  N\right)  }\right],
\label{corr6}%
\end{equation}
where
\begin{equation}
\left[  \theta^{\left(  N\right)  },\theta^{\left(  N\right)  },\overset
{2m-2}{\cdot\cdot\cdot\cdot},\theta^{\left(  N\right)  }\right]
=\omega^{i_{p_{1}}^{1},\beta_{p_{1}}^{1}}\wedge\cdot\cdot\cdot\wedge
\omega^{i_{p_{2m-2}}^{2m-2},\beta_{p_{2m-2}}^{2m-2}}\left[  X_{i_{p_{1}}%
^{1},\beta_{p_{1}}^{1}},\cdot\cdot\cdot,X_{i_{p_{2m-2}}^{2m-2},\beta
_{p_{2m-2}}^{2m-2}}\right]  \text{.}\label{corr7}%
\end{equation}
Note that the equations (\ref{t6})-(\ref{t7}) store the structure subspace
information of $\mathcal{G}$ and therefore must be mentioned if we use this
free index notation.

The following theorem generalizes theorem $2$ of Ref. \cite{deazcarraga1} to
the case of higher order Lie algebras and it establishes the conditions under
which the $1$-forms $\omega^{i_{0},\alpha_{0}},$ $\omega^{i_{1},\alpha_{1}}$
generate new higher order Lie algebras:

\begin{theorem}
Let $\left(  \mathcal{G},\left[  ,...,\right]  \right)  $ be a higher order
Lie algebras with $\mathcal{G}=V_{0}\oplus V_{1},$ where $V_{0}$ is a
submultialgebra. Let the coordinates $g^{i_{p}}$ of $G$ be rescaled by
\ $g^{i_{0}}\longrightarrow g^{i_{0}}$, $g^{i_{1}}\longrightarrow\lambda
g^{i_{1}}$. Then, the coefficient one-forms $\left\{  \omega^{i_{0},\alpha
_{0}}\right\}  $, $\left\{  \omega^{i_{1},\alpha_{1}}\right\}  $ of the
expansions (\ref{t3}) of the Maurer-Cartan forms of $\mathcal{G}^{\ast}$
determine higher order Lie algebras $\left(  \mathcal{G},(N_{0},N_{1})\left[
,...,\right]  \right)  $ when $N_{1}=N_{0}$ or $N_{1}=N_{0}+1$ of dimension
$\dim\mathcal{G}\left(  N_{0},N_{1}\right)  =(N_{0}+1)\dim V_{0}+N_{1}\dim
V_{1} $ and with structure constants (\ref{t6}).
\end{theorem}

\begin{proof}
We must prove that the set
\begin{equation}
\left\{  \omega^{i_{0},\alpha_{0}},\omega^{i_{1},\alpha_{1}}\right\}
=\left\{  \omega^{i_{0},0},\omega^{i_{0},1},\cdot\cdot\cdot,\omega
^{i_{0},N_{0}};\omega^{i_{1},0},\omega^{i_{1},1},\cdot\cdot\cdot\cdot
,\omega^{i_{1},N_{1}}\right\}, \label{t8}%
\end{equation}
it is closed for the generalized Maurer-Cartan equations (\ref{t5}) and that
the Jacobi identity is satisfied. In fact, equation (\ref{t5}) can be written
as
\begin{equation}
\tilde{d}_{m}\omega^{k_{s},\alpha_{s}}=\frac{1}{\left(  2m-2\right)  !}%
\sum\limits_{\beta_{p_{1}}^{1},\cdot\cdot\cdot,\beta_{p_{2m-2}}^{2m-2}%
=0}^{\alpha_{s}}C_{i_{p_{1}}^{1}\cdot\cdot\cdot i_{p_{2m-2}}^{2m-2}}^{k_{s}%
}\delta_{\beta_{p_{1}}^{1}+\cdot\cdot\cdot+\beta_{p_{2m-2}}^{2m-2}}%
^{\alpha_{s}}\omega^{i_{p_{1}}^{1},\beta_{p_{1}}^{1}}\wedge\cdot\cdot
\cdot\wedge\omega^{i_{p_{2m-2}}^{2m-2},\beta_{p_{2m-2}}^{2m-2}}.\label{g12}%
\end{equation}
From (\ref{g12}) we have
\begin{align}
\tilde{d}_{m}\omega^{k_{0},\alpha_{0}}&=\frac{1}{\left(  2m-2\right)  !}%
\sum\limits_{\beta_{p_{1}}^{1},\cdot\cdot\cdot,\beta_{p_{2m-2}}^{2m-2}%
=0}^{\alpha_{0}}C_{i_{p_{1}}^{1}\cdot\cdot\cdot i_{p_{2m-2}}^{2m-2}}^{k_{0}%
}\delta_{\beta_{p_{1}}^{1}+\cdot\cdot\cdot+\beta_{p_{2m-2}}^{2m-2}}%
^{\alpha_{0}}\omega^{i_{p_{1}}^{1},\beta_{p_{1}}^{1}}\wedge\cdot\cdot
\cdot\wedge\omega^{i_{p_{2m-2}}^{2m-2},\beta_{p_{2m-2}}^{2m-2}},\label{g13} \\%
\tilde{d}_{m}\omega^{k_{1},\alpha_{1}}&=\frac{1}{\left(  2m-2\right)  !}%
\sum\limits_{\beta_{p_{1}}^{1},\cdot\cdot\cdot\cdot,\beta_{p_{2m-2}}^{2m-2}%
=0}^{\alpha_{1}}C_{i_{p_{1}}^{1}\cdot\cdot\cdot i_{p_{2m-2}}^{2m-2}}^{k_{1}%
}\delta_{\beta_{p_{1}}^{1}+\cdot\cdot\cdot+\beta_{p_{2m-2}}^{2m-2}}%
^{\alpha_{1}}\omega^{i_{p_{1}}^{1},\beta_{p_{1}}^{1}}\wedge\cdot\cdot
\cdot\wedge\omega^{i_{p_{2m-2}}^{2m-2},\beta_{p_{2m-2}}^{2m-2}}.\label{g14}%
\end{align}

We now consider the forms that contribute to $\tilde{d}_{m}\omega
^{k_{s},\alpha_{s}}$:

\begin{itemize}
\item[(a)] the case $\alpha=0,$%
\[
\tilde{d}_{m}\omega^{k_{0},0}=\frac{1}{\left(  2m-2\right)  !}\sum
\limits_{\beta_{p_{1}}^{1},\cdot\cdot\cdot\cdot,\beta_{p_{2m-2}}^{2m-2}=0}%
^{0}C_{i_{p_{1}}^{1}\cdot\cdot\cdot\cdot i_{p_{2m-2}}^{2m-2}}^{k_{0}}%
\delta_{\beta_{p_{1}}^{1}+\cdot\cdot\cdot\cdot+\beta_{p_{2m-2}}^{2m-2}}%
^{0}\bullet
\]%
\begin{equation}
\bullet\omega^{i_{p_{1}}^{1},\beta_{p_{1}}^{1}}\wedge\cdot\cdot\cdot
\wedge\omega^{i_{p_{2m-2}}^{2m-2},\beta_{p_{2m-2}}^{2m-2}},\label{g15}%
\end{equation}%
\begin{equation}
\tilde{d}_{m}\omega^{k_{0},0}=\frac{1}{\left(  2m-2\right)  !}C_{i_{p_{1}}%
^{1}\cdot\cdot\cdot\cdot i_{p_{2m-2}}^{2m-2}}^{k_{0}}\omega^{i_{p_{1}}^{1}%
,0}\wedge\cdot\cdot\cdot\wedge\omega^{i_{p_{2m-2}}^{2m-2},0}.\label{g16}%
\end{equation}
The condition $\omega^{i_{p},\alpha_{p}}=0$ for $\alpha_{p}<p$ implies that,
in the sum over $p_{l}$ of equation (\ref{g16}), only terms of the form
$\omega^{i_{0}^{1}}$ survive$.$ Therefore
\begin{equation}
\tilde{d}_{m}\omega^{k_{0},0}=\frac{1}{\left(  2m-2\right)  !}C_{i_{0}%
^{1}\cdot\cdot\cdot\cdot i_{0}^{2m-2}}^{k_{0}}\omega^{i_{0}^{1},0}\wedge
\cdot\cdot\cdot\wedge\omega^{i_{0}^{2m-2},0}.\label{g17}%
\end{equation}

\item[(b)] the case $\alpha=1$%
\[
\tilde{d}_{m}\omega^{k_{0},1}=\frac{1}{\left(  2m-2\right)  !}\sum
\limits_{\beta_{p_{1}}^{1},\cdot\cdot\cdot\cdot,\beta_{p_{2m-2}}^{2m-2}=0}%
^{1}C_{i_{p_{1}}^{1}\cdot\cdot\cdot\cdot i_{p_{2m-2}}^{2m-2}}^{k_{0}}%
\delta_{\beta_{p_{1}}^{1}+\cdot\cdot\cdot\cdot+\beta_{p_{2m-2}}^{2m-2}}%
^{1}\bullet
\]%
\begin{equation}
\bullet\omega^{i_{p_{1}}^{1},\beta_{p_{1}}^{1}}\wedge\cdot\cdot\cdot
\wedge\omega^{i_{p_{2m-2}}^{2m-2},\beta_{p_{2m-2}}^{2m-2}}.\label{g18}%
\end{equation}%
\[
\tilde{d}_{m}\omega^{k_{1},1}=\frac{1}{\left(  2m-2\right)  !}\sum
\limits_{\beta_{p_{1}}^{1},\cdot\cdot\cdot\cdot,\beta_{p_{2m-2}}^{2m-2}=0}%
^{1}C_{i_{p_{1}}^{1}\cdot\cdot\cdot\cdot i_{p_{2m-2}}^{2m-2}}^{k_{1}}%
\delta_{\beta_{p_{1}}^{1}+\cdot\cdot\cdot\cdot+\beta_{p_{2m-2}}^{2m-2}}%
^{1}\bullet
\]%
\begin{equation}
\bullet\omega^{i_{p_{1}}^{1},\beta_{p_{1}}^{1}}\wedge\cdot\cdot\cdot
\wedge\omega^{i_{p_{2m-2}}^{2m-2},\beta_{p_{2m-2}}^{2m-2}}.\label{g19}%
\end{equation}

\end{itemize}
\end{proof}

\begin{itemize}
\item[(c)] the case $\alpha\geq2$%
\[
\tilde{d}_{m}\omega^{k_{0},\alpha_{0}}=\frac{1}{\left(  2m-2\right)  !}%
\sum\limits_{\beta_{p_{1}}^{1},\cdot\cdot\cdot\cdot,\beta_{p_{2m-2}}^{2m-2}%
=0}^{\alpha_{0}}C_{i_{p_{1}}^{1}\cdot\cdot\cdot\cdot i_{p_{2m-2}}^{2m-2}%
}^{k_{0}}\delta_{\beta_{p_{1}}^{1}+\cdot\cdot\cdot\cdot+\beta_{p_{2m-2}%
}^{2m-2}}^{\alpha_{0}}\bullet
\]%
\begin{equation}
\bullet\omega^{i_{p_{1}}^{1},\beta_{p_{1}}^{1}}\wedge\cdot\cdot\cdot
\wedge\omega^{i_{p_{2m-2}}^{2m-2},\beta_{p_{2m-2}}^{2m-2}}\label{g20}%
\end{equation}%
\[
\tilde{d}_{m}\omega^{k_{1},\alpha_{1}}=\frac{1}{\left(  2m-2\right)  !}%
\sum\limits_{\beta_{p_{1}}^{1},\cdot\cdot\cdot\cdot,\beta_{p_{2m-2}}^{2m-2}%
=0}^{\alpha_{1}}C_{i_{p_{1}}^{1}\cdot\cdot\cdot\cdot i_{p_{2m-2}}^{2m-2}%
}^{k_{1}}\delta_{\beta_{p_{1}}^{1}+\cdot\cdot\cdot\cdot+\beta_{p_{2m-2}%
}^{2m-2}}^{\alpha_{1}}%
\]%
\begin{equation}
\bullet\omega^{i_{p_{1}}^{1},\beta_{p_{1}}^{1}}\wedge\cdot\cdot\cdot
\wedge\omega^{i_{p_{2m-2}}^{2m-2},\beta_{p_{2m-2}}^{2m-2}}.\
\end{equation}
\ \ \ 
\end{itemize}

\bigskip

Therefore:\textbf{\ }

(1) \ For $\tilde{d}_{m}\omega^{k_{0},\alpha_{0}}$, we have

\begin{itemize}
\item[(1a)] the forms $\omega^{i_{0}^{1},\beta_{0}^{1}}$ contribute up to the
order shown in the following table:%

\[%
\begin{tabular}
[c]{||c||c||}\hline\hline
& Maximum order of $\beta_{0}^{l}$\\\hline\hline
$\alpha_{0}=0$ & $\beta_{0}^{l}\leq0$\\\hline\hline
$\alpha_{0}=1$ & $\beta_{0}^{l}\leq1$\\\hline\hline
$\alpha_{0}\geq2$ & $\beta_{0}^{l}\leq\alpha_{0}$\\\hline\hline
\end{tabular}
\ \
\]

\item[(1b)] The forms $\omega^{i_{1}^{1},\beta_{1}^{1}}$ contribute up to the
order shown in the following table:%
\[%
\begin{tabular}
[c]{||c||c||}\hline\hline
& Maximum order of $\beta_{1}^{l}$\\\hline\hline
$\alpha_{0}=0$ & there is no contribution\\\hline\hline
$\alpha_{0}=1$ & $\beta_{1}^{l}\leq1$\\\hline\hline
$\alpha_{0}\geq2$ & $\beta_{1}^{l}\leq\alpha_{0}$\\\hline\hline
\end{tabular}
\ \
\]

\item[(2)] For $\tilde{d}_{m}\omega^{k_{1},\alpha_{1}}$we have

\item[(2a)] with respect to the contribution of the forms $\omega^{i_{0}%
^{l},\beta_{0}^{l}}$ we can say

\item[(2ai)] for $\alpha_{1}=1$ we have that the maximum order of $\beta
_{0}^{l}$ can be found by analyzing the equation (\ref{g19})%
\begin{equation*}
\tilde{d}_{m}\omega^{k_{1},1}=\frac{1}{\left(  2m-2\right)  !}\sum
\limits_{\beta_{p_{1}}^{1},\cdot\cdot\cdot\cdot,\beta_{p_{2m-2}}^{2m-2}=0}%
^{1}C_{i_{p_{1}}^{1}\cdot\cdot\cdot\cdot i_{p_{2m-2}}^{2m-2}}^{k_{1}}%
\delta_{\beta_{p_{1}}^{1}+\cdot\cdot\cdot\cdot+\beta_{p_{2m-2}}^{2m-2}}%
^{1}\bullet
\end{equation*}%
\begin{equation}
\bullet\omega^{i_{p_{1}}^{1},\beta_{p_{1}}^{1}}\wedge\cdot\cdot\cdot
\wedge\omega^{i_{p_{2m-2}}^{2m-2},\beta_{p_{2m-2}}^{2m-2}}.
\end{equation}

\end{itemize}

The condition of submultialgebra, $C_{i_{0}^{1}\cdot\cdot\cdot i_{0}^{2m-2}%
}^{k_{1}}=0$ implies that in the sums on $p_{l}=\left\{  p_{1},\cdot\cdot
\cdot,p_{2m-2}\right\}  $ at least one of them, we say $p_{x} $, must be equal
to $1$. \ So that $p_{l}=0$ for $l\neq x$. This means that the condition%
\begin{equation}
\beta_{p_{1}}^{1}+\cdot\cdot\cdot\beta_{p_{x}}^{x}\cdot\cdot\cdot
+\beta_{p_{2m-2}}^{2m-2}=1,\label{g22}%
\end{equation}
takes the form%
\begin{equation}
\beta_{0}^{1}+\cdot\cdot\cdot\beta_{0}^{x-1}+\beta_{1}^{x}+\beta_{0}%
^{x+1}+\cdot\cdot\cdot+\beta_{0}^{2m-2}=1.\label{g23}%
\end{equation}

Since $\omega^{i_{p_{l}}^{l},\beta_{p_{l}}^{l}}=0$ for $\beta_{p_{l}}%
^{l}<p_{l}$ we have that to generate a non vanishing element in (\ref{g19}),
it is necessary that in the form $\omega^{i_{p_{x}}^{x},\beta_{p_{x}}^{x}%
}=\omega^{i_{1}^{x},\beta_{1}^{x}}$ \ it must be fulfilled that $\beta_{1}%
^{x}=1.$ So $p_{l}=0$ and $\beta_{p_{l}}^{l}=\beta_{0}^{l}=0$ for $l\neq x.$
\ This means that the forms $\omega^{i_{0}^{l},\beta_{0}^{l}}$ contribute to
$\tilde{d}_{m}\omega^{k_{1},\alpha_{1}}$ for $\alpha_{1}=1$ only up to the
order $\beta_{0}^{l}=0=\alpha_{1}-1.$

\begin{itemize}
\item[(2aii)] Following the same previous procedure we find that, for
$\alpha_{1}\geq2,$ the forms $\omega^{i_{0}^{l},\beta_{0}^{l}}$ contribute up
to the order $\beta_{0}^{l}\leq\alpha_{1}-1.$
\end{itemize}

The contribution of the forms $\omega^{i_{0}^{l},\beta_{0}^{l}}$ to $\tilde
{d}_{m}\omega^{k_{1},\alpha_{1}}$ is shown in the following table
\[%
\begin{tabular}
[c]{||c||c||}\hline\hline
& Maximum order of $\beta_{0}^{l}$\\\hline\hline
$\alpha_{1}=1$ & $\beta_{0}^{l}\leq0=\alpha_{1}-1$\\\hline\hline
$\alpha_{1}\geq2$ & $\beta_{0}^{l}\leq\alpha_{1}-1$\\\hline\hline
\end{tabular}
\]

\begin{itemize}
\item[(2b)] The contribution of the forms $\omega^{i_{1}^{l},\beta_{1}^{l}} $
to $\tilde{d}_{m}\omega^{k_{1},\alpha_{1}}$ is given by:
\[%
\begin{tabular}
[c]{||c||c||}\hline\hline
& Maximum order of $\beta_{1}^{l}$\\\hline\hline
$\alpha_{1}=1$ & $\beta_{1}^{l}\leq1=\alpha_{1}$\\\hline\hline
$\alpha_{1}\geq2$ & $\beta_{1}^{l}\leq\alpha_{1}$\\\hline\hline
\end{tabular}
\ \
\]
In the following table are summarized the contributions of the forms
$\omega^{i_{p_{l}}^{l},\beta_{p_{l}}^{l}}$ to $\tilde{d}_{m}\omega
^{k_{s},\alpha_{s}}:$
\end{itemize}

\ \ \ \ \ \ \ \ \ \ \ \ \ \ \ \ \ \ \ \ \ \ \ \ \ \ \ \ \ \ \ \ \ \ \ \ \ \ \ \ \ \ \ \ \ \ \ \ \ \ \ \ \ \ \ \ \ \ \ \ \ \ \ \ \
\begin{tabular}
[c]{||c||c||c||}\hline\hline
$\alpha_{s}\geq s$ & $\omega^{i_{0}^{l},\beta_{0}^{l}}$ & $\omega^{i_{1}%
^{1},\beta_{1}^{1}}$\\\hline\hline
$\tilde{d}_{m}\omega^{k_{0},\alpha_{0}}$ & $\beta_{0}^{l}\leq\alpha_{0}$ &
$\beta_{0}^{l}\leq\alpha_{0}$\\\hline\hline
$\tilde{d}_{m}\omega^{k_{1},\alpha_{1}}$ & $\beta_{0}^{l}\leq\alpha_{1}-1$ &
$\beta_{1}^{l}\leq\alpha_{1}$\\\hline\hline
\end{tabular}
\ \ \ \ \ \ \ \ \ \ \ \ \ \ \ \ \ \ \ \ \ \ \ \ \ \ \ \ \ \ \ \ \ \ \ \ \ \ \ \ \ \ \ \ \ \ \ \ \ \ \ \ \ \ \ \ \ \ \ \ \ \ \ \ \ \ \ \ 

In order that the generalized Maurer-Cartan equations be satisfied, there must
exist in (\ref{t8}) sufficient $1$-forms, so that the ($N_{0}+1)$
$\omega^{k_{0},\alpha_{0}}$ and $N_{1}$ $\omega^{k_{1},\alpha_{1}}$ (\ref{t8})
must include at least those present in its differential. This means that the
above table implies the inverse iniqualities shown in the following table

\ \ \ \ \ \ \ \ \ \ \ \ \ \ \ \ \ \ \ \ \ \ \ \ \ \ \ \ \ \ \ \ \ \ \ \ \ \ \ \ \ \ \ \ \ \ \ \ \ \ \ \ \ \ \ \ \ \ \ \ \ \ \ \
\begin{tabular}
[c]{||c||c||c||}\hline\hline
$\alpha_{s}\geq s$ & $\omega^{i_{0}^{l},\beta_{0}^{l}}$ & $\omega^{i_{1}%
^{1},\beta_{1}^{1}}$\\\hline\hline
$\tilde{d}_{m}\omega^{k_{0},\alpha_{0}}$ & $N_{0}\geq N_{0}$ & $N_{1}\geq
N_{0}$\\\hline\hline
$\tilde{d}_{m}\omega^{k_{1},\alpha_{1}}$ & $N_{0}\geq N_{1}-1$ & $N_{1}\geq
N_{1}$\\\hline\hline
\end{tabular}

The corresponding solutions to the inequations
\begin{align}
N_{1}&\geq N_{0} ,\\
N_{0}&\geq N_{1}-1,
\end{align}
are
\begin{equation}
N_{1}=N_{0},\label{g24}%
\end{equation}
or%
\begin{equation}
N_{0}=N_{1}-1.\label{g25}%
\end{equation}

These equations show the two ways in which the (\ref{t3})\ \ expansions must
be truncated. \ \ \ \ \ \ 

\section{\textbf{Dual formulation of the higher-order Lie algebra S-expansion
procedure}}

In ref. \cite{sexpansion1} was constructed an S-expansion procedure which
permits obtaining a new higher-order Lie algebra from an original one by
choosing an Abelian semigroup $S$. In the previous sections of the present
work we have generalized the expansion procedure of ref. \cite{deazcarraga1}
to the higher-order Lie algebra case. \ \ 

The $S$-expansion procedure is defined as the action of a semigroup $S$ on the
generators $T_{A}$ of the algebra, and the power series expansion is carried
out on the MC forms of the original algebra. On the other hand, the
$S$-expansion is defined on the algebra $\mathfrak{g}$\ without referring to
the group manifold, whereas the power series expansion is based on a rescaling
of the group coordinates.

It is the purpose of this section to study the $S$-expansion procedure in the
context of the group manifold and then to find the dual formulation of such an
$S$-expansion procedure.

\subsection{\textbf{S-expansion of the \ higher-order Lie algebra}}

Let's remember that the $S$-expansion method is based on combining the
structure constants of $\ \left(  \mathcal{G},\left[  ,...,\right]  \right)  $
with the inner law of a semigroup $S$ to define the Lie bracket of a new,
$S$-expanded multialgebra.

Let $S=\left\{  \lambda_{\alpha}\right\}  $ be a finite Abelian semigroup
endowed with a commutative and associative composition law $S\times
S\rightarrow S,$ $\left(  \lambda_{\alpha},\lambda_{\beta}\right)
\mapsto\lambda_{\alpha}\lambda_{\beta}=K_{\alpha\beta}^{\text{ \ \ \ \ }%
\gamma}\lambda_{\gamma}.$ The direct product $\mathfrak{G}=S\otimes
\mathcal{G}$ is defined as the cartesian product set
\begin{equation}
\mathfrak{G}=S\times\mathcal{G}=\left\{  T_{\left(  A,\alpha\right)  }%
=\lambda_{\alpha}T_{A}\text{ : }\lambda_{\alpha}\in S\text{ , }T_{A}%
\in\mathcal{G}\right\}, \label{f0}%
\end{equation}
with the composition law $\left[  ,...,\right]  _{S}:\mathfrak{G}\overset
{n}{\times...\times}\mathfrak{G}\rightarrow\mathfrak{G}$, defined by%
\begin{align}
\left[  T_{\left(  A_{1},\alpha_{1}\right)  },...,T_{\left(  A_{n},\alpha
_{n}\right)  }\right]  _{S}&=\lambda_{\alpha_{1}}...\lambda_{\alpha_{n}}\left[
T_{A_{1}},...,T_{A_{n}}\right],
\\
\left[  T_{\left(  A_{1},\alpha_{1}\right)  },...,T_{\left(  A_{n},\alpha
_{n}\right)  }\right]  _{S}&=K_{\alpha_{1}...\alpha_{n}}^{\gamma}%
C_{A_{1}...A_{n}}^{C}\lambda_{\gamma}T_{C}=C_{\left(  A_{1},\alpha_{1}\right)
...\left(  A_{n},\alpha_{n}\right)  }^{\left(  C,\gamma\right)  }T_{\left(
C,\gamma\right)  },\label{f1}%
\end{align}
where $T_{\left(  A_{i},\alpha_{i}\right)  }\in\mathfrak{G}$, $\forall
i=1,...,n,$ and \ \ $C_{\left(  A_{1},\alpha_{1}\right)  ...\left(
A_{n},\alpha_{n}\right)  }^{\left(  C,\gamma\right)  }=K_{\alpha_{1}%
...\alpha_{n}}^{\gamma}C_{A_{1}...A_{n}}^{C}.$ \ \ \ \ 

\begin{theorem}
The set $G=S\times G$ (\ref{f0}) with the composition law (\ref{f1}) defines a
new Lie multialgebra which will be called S-expanded Lie multialgebra. This
algebra is a Lie algebra structure defined over the vector space obtained by
taking $S$ copies of $G$ by means of the structure constant $C_{\left(
A_{1},\alpha_{1}\right)  ...\left(  A_{n},\alpha_{n}\right)  }^{\left(
C,\gamma\right)  }=K_{\alpha_{1}...\alpha_{n}}^{\gamma}C_{A_{1}...A_{n}}^{C}$
where $K_{\alpha_{1}...\alpha_{n}}^{\gamma}=K_{\alpha_{1}...\alpha_{n-1}%
}^{\sigma}K_{\sigma\alpha_{n}}^{\gamma}$. \ The structure constants
$C_{\left(  A_{1},\alpha_{1}\right)  ...\left(  A_{n},\alpha_{n}\right)
}^{\left(  C,\gamma\right)  }$ defined in (\ref{f1}) inherit the symmetry
properties of $C_{A_{1}...A_{n}}^{C}$ of $G$ by virtue of the abelian
character of the $S$-product.
\end{theorem}

\begin{proof}
The proof is direct and may be found in ref. \cite{sexpansion1}.
\end{proof}

\subsection{\textbf{Dual formulation of the S-expansion Procedure}}

The above theorem implies that, for every abelian semigroup $S$ and Lie
multialgebra $\mathfrak{g}$, the product $\mathfrak{G}=S\times\mathfrak{g}$ is
also a Lie multialgebra, with a Lie bracket given by eq.~(\ref{f1}). This in
turn means that it must be possible to look at this $S$-expanded Lie
multialgebra $\mathfrak{G}$\ from the dual point of view of the Maurer-Cartan
forms \cite{dualSexpansion}.

\begin{theorem}
If $S=\left\{  \lambda_{\alpha},\alpha=1,\ldots,N\right\}  $ is a finite
abelian semigroup and if $\omega^{A}$ are the Maurer-Cartan forms for a Lie
multialgebra $\mathfrak{g}$, then the Maurer-Cartan forms $\omega^{\left(
A,\alpha\right)  }$ associated with the $S$-expanded Lie multialgebra
$\mathfrak{G}=S\times\mathfrak{g}$ [cf. Theorem~1] are related to the
$\omega^{A}$ by
\begin{equation}
\omega^{A}=%
{\displaystyle\sum\limits_{\lambda_{\alpha}\in S}}
\lambda_{\alpha}\omega^{\left(  A,\alpha\right)  },\label{cinco}%
\end{equation}
and satisfy the generalized Maurer Cartan equations
\begin{equation}
\tilde{d}_{m}\omega^{\left(  A,\alpha\right)  }=\frac{1}{\left(  2m-2\right)
!}C_{\left(  B_{1},\beta_{1}\right)  \cdot\cdot\cdot\left(  B_{2m-2}%
,\beta_{2m-2}\right)  }%
^{\ \ \ \ \ \ \ \ \ \ \ \ \ \ \ \ \ \ \ \ \ \left(
A,\alpha\right)  }\omega^{\left(  B_{1},\beta_{1}\right)  }\cdot\cdot
\cdot\omega^{\left(  B_{2m-2},\beta_{2m-2}\right)  }.\label{seis}%
\end{equation}

\end{theorem}

\begin{proof}
Introducing eq. \ (\ref{cinco})\ into the generalized Maurer-Cartan equations%
\begin{equation}
\tilde{d}_{m}\omega^{A}=\frac{1}{\left(  2m-2\right)  !}C_{B_{1}\cdot
\cdot\cdot\cdot\cdot\cdot B_{2m-2}}%
^{\ \ \ \ \ \ \ \ \ \ \ \ \ \text{\ \ \ \ \ \ }A}\omega^{B_{1}}\cdot\cdot
\cdot\cdot\cdot\omega^{B_{2m-2}},\label{v4}%
\end{equation}
we obtain
\begin{equation}
\tilde{d}_{m}\omega^{\left(  A,\alpha\right)  }=\frac{1}{\left(  2m-2\right)
!}\sum_{\beta_{1},...,\beta_{2m-2}}C_{\left(  B_{1},\beta_{1}\right)
....\left(  B_{2m-2},\beta_{2m-2}\right)  }%
^{\ \ \ \ \ \ \ \ \ \ \ \ \ \ \ \ \ \ \ \ \ \ \ \ \left(
A,\alpha\right)  }\omega^{\left(  B_{1},\beta_{1}\right)  }...\omega^{\left(
B_{2m-2},\beta_{2m-2}\right)  },\label{v6}%
\end{equation}
where $\Omega_{\left(  B_{1},\beta_{1}\right)  ......\left(  B_{2m-2}%
,\beta_{2m-2}\right)  }^{\ \ \ \ \ \ \ \ \ \ \ \ \ \left(  A,\alpha\right)
}=\Omega_{B_{1}......B_{2m-2}}^{\ \ \ \ \ \ \ \ \ \ \ \ \ A}K_{\beta
_{1}...\beta_{2m-2}}^{\alpha}$. \ Using the sum convention, equation
(\ref{v6})\ can be written as%

\begin{equation}
\tilde{d}_{m}\omega^{\left(  A,\alpha\right)  }=\frac{1}{\left(  2m-2\right)
!}C_{\left(  B_{1},\beta_{1}\right)  ...\left(  B_{2m-2},\beta_{2m-2}\right)
}^{\ \ \ \ \ \ \ \ \ \ \ \ \ \left(  A,\alpha\right)  }\omega^{\left(
B_{1},\beta_{1}\right)  }...\omega^{\left(  B_{2m-2},\beta_{2m-2}\right)
}.\label{v7}%
\end{equation}
\ This concludes the proof.
\end{proof}%

In the compact notation that uses the canonical one-form 
$\theta^{\left(  N\right)  }\omega^{\left(  A,\alpha\right)  }X_{\left(
A,\alpha\right)  },$ the eq. (\ref{v7}) can be written as

\begin{equation}
\tilde{d}_{m}\theta^{\left(  S\right)  }=\frac{1}{\left(  2m-2\right)
!}\left[  \theta^{\left(  S\right)  },\theta^{\left(  S\right)  }%
,\overset{2m-2}{\cdot\cdot\cdot\cdot},\theta^{\left(  S\right)  }\right],
\end{equation}
where
\begin{equation}
\left[  \theta^{\left(  S\right)  },\theta^{\left(  S\right)  },\overset
{2m-2}{\cdot\cdot\cdot\cdot},\theta^{\left(  S\right)  }\right]
=\omega^{\left(  B_{1},\beta_{1}\right)  }\wedge\cdot\cdot\cdot\wedge
\omega^{\left(  B_{2m-2},\beta_{2m-2}\right)  }\left[  X_{\left(  B_{1}%
,\beta_{1}\right)  },\cdot\cdot\cdot,X_{\left(  B_{2m-2},\beta_{2m-2}\right)
}\right]  \text{.}%
\end{equation}

It is perhaps interesting to notice that the relation shown in
eq.~(\ref{cinco}) is analogous to the method of power series expansion
developed in Ref.~\cite{deazcarraga1} and in the above sections.

\subsection{$0_{S}$\textbf{-Reduction of }$S$\textbf{-expanded Lie Algebras}}

Now we present the dual formulation for the $0_{S}$-reduction of an
$S$-expanded Lie multialgebra $\mathfrak{G}$, formulated in the language of
the MC forms.

Let $S=\left\{  \lambda_{i},i=1,\ldots,N\right\}  \cup\left\{  \lambda
_{N+1}=0_{S}\right\}  $ be an abelian semigroup with zero. The expanded
Maurer-Cartan forms $\omega^{\left(  A,\alpha\right)  }$ are then given by
\begin{equation}
\omega^{A}=\sum_{i=1}^{N}\lambda_{i}\omega^{\left(  A,i\right)  }+0_{S}%
\tilde{\omega}^{A},\label{diecisiete}%
\end{equation}
where $\tilde{\omega}^{A}=\omega^{\left(  A,N+1\right)  }$. We shall show that
the Maurer Cartan forms $\omega^{\left(  A,i\right)  }$ by themselves (without
including $\tilde{\omega}^{A}$) are those of a Lie multialgebra-the $0_{S}%
$-reduced multialgebra $\mathfrak{G}_{R}$.

It can be shown~\cite{sexpansion1} that $C_{\left(  A_{1},i_{1}\right)
...\left(  A_{n},i_{n}\right)  }^{\left(  C,k\right)  }=K_{i_{1}\cdot
\cdot\cdot i_{n}}^{k}C_{A_{1}...A_{n}}^{C}$ are the structure constants for
the $0_{S}$-reduced $S$-expanded multialgebra $\mathfrak{G}_{R}$, which is
generated by $T_{(A,i)}$:%
\begin{equation}
\left[  T_{\left(  A_{1},i_{1}\right)  },...,T_{\left(  A_{n},i_{n}\right)
}\right]  _{S}=K_{i_{1}\cdot\cdot\cdot\cdot i_{n}}^{k}C_{A_{1}...A_{n}}%
^{C}T_{\left(  C,k\right)  }.\label{dieciocho}%
\end{equation}

The following Theorem~gives the equivalent statement in terms of Maurer-Cartan
forms (see \cite{dualSexpansion}):
\begin{theorem}
Let $S=\left\{  \lambda_{i},i=1,\ldots,N\right\}  \cup\left\{  \lambda
_{N+1}=0_{S}\right\}  $ be an abelian semigroup with zero and let $\left\{
\omega^{\left(  A,i\right)  },i=1,\ldots,N\right\}  \cup\left\{
\omega^{\left(  A,N+1\right)  }=\tilde{\omega}^{A}\right\}  $ be the MC forms
for the $S$-expanded multialgebra $\mathfrak{G}=S\times\mathfrak{g}$ of
$\mathfrak{g}$\ by the semigroup $S$. Then, $\left\{  \omega^{\left(
A,i\right)  },i=1,\ldots,N\right\}  $ are the Maurer-Cartan forms for the
$0_{S}$-reduced $S$-expanded multialgebra $\mathfrak{G}_{R}$.
\end{theorem}

\begin{proof}
The Maurer-Cartan forms for the S-expanded multialgebra $\mathfrak{G}$ satisfy
the generalized Maurer Cartan equations $\left[  \text{cf. eq. (\ref{seis}%
)}\right]  $
\begin{equation}
\tilde{d}_{m}\omega^{\left(  A,\alpha\right)  }=\frac{1}{\left(  2m-2\right)
!}C_{\left(  B_{1},\beta_{1}\right)  \cdot\cdot\cdot\left(  B_{2m-2}%
,\beta_{2m-2}\right)  }%
^{\ \ \ \ \ \ \ \ \ \ \ \ \ \ \ \ \ \ \ \ \ \ \ \ \ \ \ \ \left(
A,\alpha\right)  }\omega^{\left(  B_{1},\beta_{1}\right)  }\cdot\cdot
\cdot\omega^{\left(  B_{2m-2},\beta_{2m-2}\right)  }.\label{diecinueve}%
\end{equation}
Introducing (\ref{diecisiete}) into (\ref{diecinueve}) we have%
\[
\sum_{i=1}^{N}\lambda_{i}\tilde{d}_{m}\omega^{\left(  A,i\right)  }%
+0_{S}\tilde{d}_{m}\omega^{\left(  A,N+1\right)  }%
\]%
\begin{equation}
=\frac{1}{\left(  2m-2\right)  !}C_{B_{1}....B_{2m-2}}%
^{\ \ \ \ \ \ \ \ \ \ \ \ \ A}\left[  \left(
\begin{array}
[c]{c}%
\sum_{j_{1}}^{N}\lambda_{j_{1}}\omega^{\left(  B_{1},j_{1}\right)  }\\
+0_{S}\omega^{\left(  B_{1},N+1\right)  }%
\end{array}
\right)  \times\cdot\cdot\cdot\times\left(
\begin{array}
[c]{c}%
\sum_{j_{2m-2}}^{N}\lambda_{j_{2m-2}}\omega^{\left(  B_{2m-2},j_{2m-2}\right)
}\\
+0_{S}\omega^{\left(  B_{2m-2},N+1\right)  }%
\end{array}
\right)  \right]  .
\end{equation}
On the other hand we can write%
\[
\sum_{\alpha=1}^{N+1}\lambda_{\alpha}\tilde{d}_{m}\omega^{\left(
A,\alpha\right)  }=\frac{1}{\left(  2m-2\right)  !}C_{B_{1}......B_{2m-2}}%
^{\ \ \ \ \ \ \ \ \ \ \ \ \ \text{\ \ \ \ \ \ }A}\left(  \sum_{\beta_{1}%
}^{N+1}\lambda_{\beta_{1}}\omega^{\left(  B_{1},\beta_{1}\right)  }\right)
...\left(  \sum_{\beta_{2m-2}}^{N+1}\lambda_{\beta_{2m-2}}\omega^{\left(
B_{2m-2},\beta_{2m-2}\right)  }\right),
\]%
\begin{equation}
=\sum_{\alpha}^{N+1}\lambda_{\alpha}\left(  \frac{1}{\left(  2m-2\right)
!}\sum_{\beta_{1},...,\beta_{2m-2}}^{N+1}C_{\left(  B_{1},\beta_{1}\right)
......\left(  B_{2m-2},\beta_{2m-2}\right)  }%
^{\ \ \ \ \ \ \ \ \ \ \ \ \ \text{\ \ \ \ \ \ \ \ \ \ \ \ \ \ \ \ \ \ \ \ \ \ \ \ \ \ }%
\left(  A,\alpha\right)  }\omega^{\left(  B_{1},\beta_{1}\right)  }%
...\omega^{\left(  B_{2m-2},\beta_{2m-2}\right)  }\right)  .
\end{equation}
Since
\begin{equation}
\sum_{\alpha}^{N+1}\lambda_{\alpha}\tilde{d}_{m}\omega^{\left(  A,\alpha
\right)  }=\sum_{i=1}^{N}\lambda_{i}\tilde{d}_{m}\omega^{\left(  A,i\right)
}+0_{S}\tilde{d}_{m}\omega^{\left(  A,N+1\right)  },%
\end{equation}%
we have
\[
\sum_{\alpha}^{N+1}\lambda_{\alpha}\left(  \frac{1}{\left(  2m-2\right)
!}\sum_{\beta_{1},...,\beta_{2m-2}}^{N+1}C_{\left(  B_{1},\beta_{1}\right)
......\left(  B_{2m-2},\beta_{2m-2}\right)  }%
^{\ \ \ \ \ \ \ \ \ \ \ \ \ \left(  A,\alpha\right)  }\omega^{\left(
B_{1},\beta_{1}\right)  }...\omega^{\left(  B_{2m-2},\beta_{2m-2}\right)
}\right)
\]%
\[
=\sum_{i}^{N}\lambda_{i}\left(  \frac{1}{\left(  2m-2\right)  !}\sum
_{i_{1},...,i_{2m-2}}^{N}C_{\left(  B_{1},i_{1}\right)  ......\left(
B_{2m-2},i_{2m-2}\right)  }%
^{\ \ \ \ \ \ \ \ \ \ \ \ \ \text{\ \ \ \ \ \ \ \ \ \ \ \ \ \ \ \ \ \ \ \ \ \ \ \ \ }%
\left(  A,i\right)  }\omega^{\left(  B_{1},i_{1}\right)  }...\omega^{\left(
B_{2m-2},i_{2m-2}\right)  }\right)
\]%
\begin{equation}
+0_{S}\left(  \frac{1}{\left(  2m-2\right)  !}\sum_{\beta_{1},...,\beta
_{2m-2}}^{N+1}C_{\left(  B_{1},\beta_{1}\right)  ......\left(  B_{2m-2}%
,\beta_{2m-2}\right)  }^{\ \ \ \ \ \ \ \ \ \ \ \ \ \left(  A,N+1\right)
}\omega^{\left(  B_{1},\beta_{1}\right)  }...\omega^{\left(  B_{2m-2}%
,\beta_{2m-2}\right)  }\right)  ,
\end{equation}
we have that the generalized Maurer-Cartan equations takes the form%
\[
\left(  \sum_{i=1}^{N}\lambda_{\alpha}\tilde{d}_{m}\omega^{\left(
A,\alpha\right)  }+0_{S}\tilde{d}_{m}\omega^{\left(  A,N+1\right)  }\right)
\]%
\[
=\sum_{i}^{N}\lambda_{i}\left(  \frac{1}{\left(  2m-2\right)  !}\sum
_{i_{1},...,i_{2m-2}}^{N}C_{\left(  B_{1},i_{1}\right)  ......\left(
B_{2m-2},i_{2m-2}\right)  }%
^{\ \ \ \ \ \ \ \ \ \ \ \ \ \ \ \ \ \ \ \ \ \ \ \ \ \ \ \ \ \ \ \ \ \ \ \left(
A,i\right)  }\omega^{\left(  B_{1},i_{1}\right)  }...\omega^{\left(
B_{2m-2},i_{2m-2}\right)  }\right)
\]%
\begin{equation}
+0_{S}\left(  \frac{1}{\left(  2m-2\right)  !}\sum_{\beta_{1},...,\beta
_{2m-2}}^{N+1}C_{\left(  B_{1},\beta_{1}\right)  ......\left(  B_{2m-2}%
,\beta_{2m-2}\right)  }^{\ \ \ \ \ \ \ \ \ \ \ \ \ \left(  A,N+1\right)
}\omega^{\left(  B_{1},\beta_{1}\right)  }...\omega^{\left(  B_{2m-2}%
,\beta_{2m-2}\right)  }\right)  .
\end{equation}
So that,%
\begin{equation}
\tilde{d}_{m}\omega^{\left(  A,i\right)  }=\frac{1}{\left(  2m-2\right)
!}\sum_{i_{1},...,i_{2m-2}}^{N}C_{\left(  B_{1},i_{1}\right)  ......\left(
B_{2m-2},i_{2m-2}\right)  }%
^{\ \ \ \ \ \ \ \ \ \ \ \ \ \ \ \ \ \ \ \ \ \ \ \ \ \ \ \ \ \ \ \ \ \ \ \ \left(
A,i\right)  }\omega^{\left(  B_{1},i_{1}\right)  }...\omega^{\left(
B_{2m-2},i_{2m-2}\right)  },%
\end{equation}%
\begin{equation}
\tilde{d}_{m}\omega^{\left(  A,N+1\right)  }=\frac{1}{\left(  2m-2\right)
!}\sum_{\beta_{1},...,\beta_{2m-2}}^{N+1}C_{\left(  B_{1},\beta_{1}\right)
......\left(  B_{2m-2},\beta_{2m-2}\right)  }%
^{\ \ \ \ \ \ \ \ \ \ \ \ \ \ \ \ \ \ \ \ \ \ \ \ \ \ \ \ \left(
A,N+1\right)  }\omega^{\left(  B_{1},\beta_{1}\right)  }...\omega^{\left(
B_{2m-2},\beta_{2m-2}\right)  }.
\end{equation}
Applying the so-called $0_{s}$-reduction we obtain%
\begin{equation}
\tilde{d}_{m}\omega^{\left(  A,i\right)  }=\frac{1}{\left(  2m-2\right)
!}\sum_{i_{1},...,i_{2m-2}}^{N}C_{\left(  B_{1},i_{1}\right)  ......\left(
B_{2m-2},i_{2m-2}\right)  }%
^{\ \ \ \ \ \ \ \ \ \ \ \ \ \ \ \ \ \ \ \ \ \ \ \ \left(
A,i\right)  }\omega^{\left(  B_{1},i_{1}\right)  }...\omega^{\left(
B_{2m-2},i_{2m-2}\right)  }.\label{tak}%
\end{equation}
This concludes the proof.
\end{proof}

In the compact notation that uses the canonical one-form 
$\theta^{\left(  S\right)  }=\omega^{\left(  A,i\right)  }X_{\left(
A,i\right)  }$, $i=1,...,N$, the eq. (\ref{tak}) can be written as
\begin{equation}
\tilde{d}_{m}\theta^{\left(  S\right)  }=\frac{1}{\left(  2m-2\right)
!}\left[  \theta^{\left(  S\right)  },\theta^{\left(  S\right)  }%
,\overset{2m-2}{\cdot\cdot\cdot\cdot},\theta^{\left(  S\right)  }\right]
\text{,}%
\end{equation}
\begin{equation}
\left[  \theta^{\left(  S\right)  },\theta^{\left(  S\right)  },\overset
{2m-2}{\cdot\cdot\cdot\cdot},\theta^{\left(  S\right)  }\right]
=\omega^{\left(  B_{1},i_{1}\right)  }\wedge\cdot\cdot\cdot\wedge
\omega^{\left(  B_{2m-2},i_{2m-2}\right)  }\left[  X_{\left(  B_{1}%
,i_{1}\right)  },\cdot\cdot\cdot,X_{\left(  B_{2m-2},i_{2m-2}\right)
}\right]  \text{.}%
\end{equation}

\subsection{Resonant submultialgebras}

From ref. \cite{sexpansion}, \cite{sexpansion1} we known that if
$\mathcal{G}=\oplus_{p\in I}V_{p}$ is a decompostion of $\mathcal{G}$\ \ into
subspaces $V_{p}$, with a structure described by the subsets $i_{\left(
p_{1},...,p_{n}\right)  }\subset I$ such that
\begin{equation}
\left[  V_{p_{1}},...,V_{p_{n}}\right]  \subset%
{\displaystyle\bigoplus\limits_{r\in i_{\left(  p_{1},...,p_{n}\right)  }}}
V_{r},\label{sm1}%
\end{equation}
and if $S=\cup_{p\in I}S_{p}$ is a subset decomposition of the Abelian
semigroup $S$ such that%
\begin{equation}
S_{p_{1}}\times S_{p_{2}}\times\cdot\cdot\cdot\cdot\times S_{p_{n}}\subset%
{\displaystyle\bigcap\limits_{r\in i_{\left(  p_{1},...,p_{n}\right)  }}}
S_{r},\label{sm2}%
\end{equation}
then we say that this decomposition is in resonance with the subspace
decomposition of $\mathcal{G}=\oplus_{p\in I}V_{p}$.

In the same refs. \cite{sexpansion}, \cite{sexpansion1} it was shown that if
$\mathcal{G}=\oplus_{p\in I}V_{p}$ is a decompostion of $\mathcal{G}$\ \ into
subspaces $V_{p}$ with a structure described by \ $\left[  V_{p_{1}%
},...,V_{p_{n}}\right]  \subset%
{\displaystyle\bigoplus\limits_{r\in i_{\left(  p_{1},...,p_{n}\right)  }}}
V_{r}$ and if $S=\cup_{p\in I}S_{p}$ is a subset decomposition of the Abelian
semigroup $S$ with the structure given by $S_{p_{1}}\times S_{p_{2}}%
\times\cdot\cdot\cdot\cdot\times S_{p_{n}}\subset%
{\displaystyle\bigcap\limits_{r\in i_{\left(  p_{1},...,p_{n}\right)  }}}
S_{r}$, then the algebra given by \ $\mathfrak{G}_{R}=%
{\displaystyle\bigoplus\limits_{p\in I}}
S_{p}\otimes V_{p}=%
{\displaystyle\bigoplus\limits_{p\in I}}
W_{p}$ is a subalgebra of the S-expanded multialgebra called a resonant submultialgebra.

If $\left\{  T_{a_{p}}\right\}  $ denote the basis of $V_{p}$ , $\lambda
_{\alpha_{q}}\in S_{q}$ and if $T_{\left(  a_{p},\alpha_{q}\right)  }%
=\lambda_{\alpha_{q}}T_{a_{p}}$ then we can write
\begin{equation}
\left[  T_{\left(  a_{p_{1}}^{1},\alpha_{q_{1}}^{1}\right)  },...,T_{\left(
a_{p_{n}}^{n},\alpha_{q_{n}}^{n}\right)  }\right]  _{S}=C_{\left(  a_{p_{1}%
}^{1},\alpha_{q_{1}}^{1}\right)  ...\left(  a_{p_{n}}^{n},\alpha_{q_{n}}%
^{n}\right)  }^{\ \ \ \ \ \ \ \ \ \ \ \ \ \ \ \ \ \ \ \ \ \ \ \ \ \ \left(
c_{s},\gamma_{r}\right)  }T_{\left(  c_{s},\gamma_{r}\right)  }\text{,}%
\label{sm3}%
\end{equation}
which means that the structure constants of the resonant submultialgebra are
given by%

\begin{equation}
C_{\left(  a_{p_{1}}^{1},\alpha_{p_{1}}^{1}\right)  ...\left(  a_{p_{n}}%
^{n},\alpha_{p_{n}}^{n}\right)  }%
^{\ \ \ \ \ \ \ \ \ \ \ \ \ \ \ \ \ \ \ \ \ \ \ \ \ \ \left(  c_{r},\gamma
_{r}\right)  }=K_{\alpha_{p_{1}}^{1}...\alpha_{p_{n}}^{n}}%
^{\ \ \ \ \ \ \ \ \ \ \gamma_{r}}C_{a_{p_{1}}^{1}...a_{p_{n}}^{n}%
}^{\ \ \ \ \ \ \ \ \ \ \ c_{r}}\text{.}%
\end{equation}

The following theorem provides the Maurer-Cartan equations for the
resonant submultialgebra:
\begin{theorem}
Let $\left\{  \omega^{a_{p}}\right\}  $ be a basis of $V_{p}^{\ast}$ and let
$\lambda_{\alpha_{q}}\in S_{q}$. Then%
\begin{equation}
\omega^{a_{p}}=\sum_{\lambda_{\alpha_{p}}\in S_{p}}\lambda_{\alpha_{p}}%
\omega^{\left(  a_{p},\alpha_{p}\right)  },\label{sm9}%
\end{equation}
and the Maurer-Cartan equations for the resonant submultialgebra of the
$S$-expanded multialgebra are given by $S$%
\begin{equation}
\tilde{d}_{m}\omega^{\left(  c_{r},\gamma_{r}\right)  }=\frac{1}{\left(
2m-2\right)  !}C_{\left(  a_{p_{1}}^{1},\alpha_{p_{1}}^{1}\right)  ...\left(
a_{p_{2m-2}}^{2m-2},\alpha_{p_{2m-2}}^{2m-2}\right)  }%
^{\ \ \ \ \ \ \ \ \ \ \ \ \ \left(  c_{r},\gamma_{r}\right)  }\omega^{\left(
a_{p_{1}}^{1},\alpha_{p_{1}}^{1}\right)  }...\omega^{\left(  a_{p_{2m-2}%
}^{2m-2},\alpha_{p_{2m-2}}^{2m-2}\right)  },\label{sm7}%
\end{equation}
where
\begin{equation}
C_{\left(  a_{p_{1}}^{1},\alpha_{p_{1}}^{1}\right)  ...\left(  a_{p_{2m-2}%
}^{2m-2},\alpha_{p_{2m-2}}^{2m-2}\right)  }^{\ \ \ \ \ \ \ \ \ \ \ \ \ \left(
c_{s},\gamma_{s}\right)  }=K_{\alpha_{p_{1}}^{1}...\alpha_{p_{2m-2}}^{2m-2}%
}^{\gamma_{s}}C_{a_{p_{1}}^{1}...a_{p_{2m-2}}^{2m-2}}^{c_{s}}\ \text{,}%
\ \ \text{with }r,p_{i}\in I\text{.}%
\end{equation}

\end{theorem}

\begin{proof}
\bigskip\ The generalized Maurer-Cartan equations are given by
\begin{equation}
\tilde{d}_{m}\omega^{A}\left(  g,\lambda\right)  =\frac{1}{\left(
2m-2\right)  !}C_{B_{1}...B_{2m-2}}^{\ \ \ \ \ \ \ \ \ \ \ \ \ A}\omega
^{B_{1}}\left(  g,\lambda\right)  ...\omega^{B_{2m-2}}\left(  g,\lambda
\right)  .
\end{equation}
Introducing%
\begin{equation}
\omega^{a_{p}}=\sum_{\lambda_{\alpha_{p}}\in S_{p}}\lambda_{\alpha_{p}}%
\omega^{\left(  a_{p},\alpha_{p}\right)  },%
\end{equation}
into the generalized Maurer-Cartan equations, we have \ \
\begin{align}
&  \sum_{\lambda_{\gamma_{s}}\in S_{s}}\lambda_{\gamma_{s}}\tilde{d}_{m}%
\omega^{\left(  c_{s},\gamma_{s}\right)  } \nonumber \label{ap1}\\
&  =\frac{1}{\left(  2m-2\right)  !}C_{a_{p_{1}}^{1}...a_{p_{2m-2}}^{2m-2}%
}^{\ \ \ \ \ \ \ \ \ \ \ \ \ c_{s}}\left(  \sum_{\alpha_{p_{1}}^{1}}%
\lambda_{\alpha_{p_{1}}^{1}}\omega^{\left(  a_{p_{1}}^{1},\alpha_{p_{1}}%
^{1}\right)  }\right)  ...\left(  \sum_{\alpha_{p_{2m-2}}^{2m-2}}%
\lambda_{\alpha_{p_{2m-2}}^{2m-2}}\omega^{\left(  a_{p_{2m-2}}^{2m-2}%
,\alpha_{p_{2m-2}}^{2m-2}\right)  }\right) \nonumber\\
&  =\frac{1}{\left(  2m-2\right)  !}C_{a_{p_{1}}^{1}...a_{p_{2m-2}}^{2m-2}%
}^{\ \ \ \ \ \ \ \ \ \ \ \ \ c_{s}}\sum_{\alpha_{p_{1}}^{1},...,\alpha
_{p_{2m-2}}^{2m-2}}\lambda_{\alpha_{p_{1}}^{1}}...\lambda_{\alpha_{p_{2m-2}%
}^{2m-2}}\omega^{\left(  a_{p_{1}}^{1},\alpha_{p_{1}}^{1}\right)  }%
...\omega^{\left(  a_{p_{2m-2}}^{2m-2},\alpha_{p_{2m-2}}^{2m-2}\right)
}.
\end{align}
From the generalized resonance condition, we have
\begin{equation}
\lambda_{\alpha_{p_{1}}^{1}}...\lambda_{\alpha_{p_{2m-2}}^{2m-2}}%
=K_{\alpha_{p_{1}}^{1}...\alpha_{p_{2m-2}}^{2m-2}}^{\tilde{\gamma}}%
\lambda_{\tilde{\gamma}}\text{ donde }\lambda_{\tilde{\gamma}}\in\tilde
{S}\left(  p_{1},...,p_{2m-2}\right)  =%
{\displaystyle\bigcap\limits_{t\in i\left(  p_{1},...,p_{2m-2}\right)  }}
S_{t}.
\end{equation}
Since the condition $C_{a_{p_{1}}^{1}...a_{p_{2m-2}}^{2m-2}}%
^{\ \ \ \ \ \ \ \ \ \ \ \ \ c_{s}}\neq0$ implies $s\in i_{\left(
p_{1},...,p_{2m-2}\right)  } $, we have%
\begin{equation}
\tilde{S}\left(  p_{1},...,p_{2m-2}\right)  =%
{\displaystyle\bigcap\limits_{t\in i_{\left(  p_{1},...,p_{2m-2}\right)  }}}
S_{t}\subseteq S_{s}.
\end{equation}
This means that if $\tilde{S}\subseteq S_{s}$ then we can write%
\begin{equation}
\lambda_{\alpha_{p_{1}}^{1}}...\lambda_{\alpha_{p_{2m-2}}^{2m-2}}%
=K_{\alpha_{p_{1}}^{1}...\alpha_{p_{2m-2}}^{2m-2}}^{\gamma_{s}}\lambda
_{\gamma_{s}}\text{ where }\lambda_{\gamma_{s}}\in S_{s}.
\end{equation}
Introducing these results into (\ref{ap1}) we have%

\begin{align}
&  \sum_{\lambda_{\gamma_{s}}\in S_{s}}\lambda_{\gamma_{s}}\tilde{d}_{m}%
\omega^{\left(  c_{s},\gamma_{s}\right)  } \nonumber \\
&  =\frac{1}{\left(  2m-2\right)  !}C_{a_{p_{1}}^{1}...a_{p_{2m-2}}^{2m-2}%
}^{\ \ \ \ \ \ \ \ \ \ \ \ \ c_{s}}\sum_{\alpha_{p_{1}}^{1},...,\alpha
_{p_{2m-2}}^{2m-2}}\sum_{\lambda_{\gamma_{s}}\in S_{s}}K_{\alpha_{p_{1}}%
^{1}...\alpha_{p_{2m-2}}^{2m-2}}^{\gamma_{s}}\lambda_{\gamma_{s}}%
\omega^{\left(  a_{p_{1}}^{1},\alpha_{p_{1}}^{1}\right)  }...\omega^{\left(
a_{p_{2m-2}}^{2m-2},\alpha_{p_{2m-2}}^{2m-2}\right)  }. \nonumber
\end{align}%
\begin{align*}
&  \sum_{\lambda_{\gamma_{s}}\in S_{s}}\lambda_{\gamma_{s}}\tilde{d}_{m}%
\omega^{\left(  c_{s},\gamma_{s}\right)  }\\
&  =\frac{1}{\left(  2m-2\right)  !}\sum_{\lambda_{\gamma_{s}}\in S_{s}%
}\lambda_{\gamma_{s}}\sum_{\alpha_{p_{1}}^{1},...,\alpha_{p_{2m-2}}^{2m-2}%
}C_{a_{p_{1}}^{1}...a_{p_{2m-2}}^{2m-2}}^{\ \ \ \ \ \ \ \ \ \ \ \ \ c_{s}%
}K_{\alpha_{p_{1}}^{1}...\alpha_{p_{2m-2}}^{2m-2}}^{\gamma_{s}}\omega^{\left(
a_{p_{1}}^{1},\alpha_{p_{1}}^{1}\right)  }...\omega^{\left(  a_{p_{2m-2}%
}^{2m-2},\alpha_{p_{2m-2}}^{2m-2}\right)  }.
\end{align*}%
\begin{align*}
&  \sum_{\lambda_{\gamma_{s}}\in S_{s}}\lambda_{\gamma_{s}}\tilde{d}_{m}%
\omega^{\left(  c_{s},\gamma_{s}\right)  }\\
&  =\sum_{\lambda_{\gamma_{s}}\in S_{s}}\lambda_{\gamma_{s}}\left(  \frac
{1}{\left(  2m-2\right)  !}\sum_{\alpha_{p_{1}}^{1},...,\alpha_{p_{2m-2}%
}^{2m-2}}C_{\left(  a_{p_{1}}^{1},\alpha_{p_{1}}^{1}\right)  ...\left(
a_{p_{2m-2}}^{2m-2},\alpha_{p_{2m-2}}^{2m-2}\right)  }%
^{\ \ \ \ \ \ \ \ \ \ \ \ \ \left(  c_{s},\gamma_{s}\right)  }\omega^{\left(
a_{p_{1}}^{1},\alpha_{p_{1}}^{1}\right)  }...\omega^{\left(  a_{p_{2m-2}%
}^{2m-2},\alpha_{p_{2m-2}}^{2m-2}\right)  }\right)  .
\end{align*}
Therefore the generalized Maurer-Cartan equations for the resonant
submultialgebra are given by
\begin{equation}
\tilde{d}_{m}\omega^{\left(  c_{r},\gamma_{r}\right)  }=\frac{1}{\left(
2m-2\right)  !}\sum_{\alpha_{p_{1}}^{1},...,\alpha_{p_{2m-2}}^{2m-2}%
}C_{\left(  a_{p_{1}}^{1},\alpha_{p_{1}}^{1}\right)  ...\left(  a_{p_{2m-2}%
}^{2m-2},\alpha_{p_{2m-2}}^{2m-2}\right)  }^{\ \ \ \ \ \ \ \ \ \ \ \ \ \left(
c_{s},\gamma_{s}\right)  }\omega^{\left(  a_{p_{1}}^{1},\alpha_{p_{1}}%
^{1}\right)  }...\omega^{\left(  a_{p_{2m-2}}^{2m-2},\alpha_{p_{2m-2}}%
^{2m-2}\right)  },%
\end{equation}
which it can written in the form,%
\begin{equation}
\tilde{d}_{m}\omega^{\left(  c_{r},\gamma_{r}\right)  }=\frac{1}{\left(
2m-2\right)  !}C_{\left(  a_{p_{1}}^{1},\alpha_{p_{1}}^{1}\right)  ...\left(
a_{p_{2m-2}}^{2m-2},\alpha_{p_{2m-2}}^{2m-2}\right)  }%
^{\ \ \ \ \ \ \ \ \ \ \ \ \ \ \ \ \ \ \ \ \ \ \ \ \ \ \ \ \ \ \ \ \ \ \ \ \ \left(
c_{r},\gamma_{r}\right)  }\omega^{\left(  a_{p_{1}}^{1},\alpha_{p_{1}}%
^{1}\right)  }...\omega^{\left(  a_{p_{2m-2}}^{2m-2},\alpha_{p_{2m-2}}%
^{2m-2}\right)  }.
\end{equation}

\end{proof}

In the compact notation we have
\begin{equation}
\tilde{d}_{m}\theta^{\left(  S\right)  }=\frac{1}{\left(  2m-2\right)
!}\left[  \theta^{\left(  S\right)  },\theta^{\left(  S\right)  }%
,\overset{2m-2}{\cdot\cdot\cdot\cdot},\theta^{\left(  S\right)  }\right]
\text{,}%
\end{equation}
where
\begin{align}
\mathbf{\ }\theta^{\left(  S\right)  }&=\omega^{\left(  c_{r},\gamma
_{r}\right)  }X_{\left(  c_{r},\gamma_{r}\right)  } \\
\left[  \theta^{\left(  S\right)  },\theta^{\left(  S\right)  },\overset
{2m-2}{\cdot\cdot\cdot\cdot},\theta^{\left(  S\right)  }\right]
&=\omega^{\left(  a_{p_{1}}^{1},\alpha_{p_{1}}^{1}\right)  }\wedge\cdot
\cdot\cdot\wedge\omega^{\left(  a_{p_{2m-2}}^{2m-2},\alpha_{p_{2m-2}}%
^{2m-2}\right)  }\left[  X_{\left(  a_{p_{1}}^{1},\alpha_{p_{1}}^{1}\right)
},\cdot\cdot\cdot,X_{\left(  a_{p_{2m-2}}^{2m-2},\alpha_{p_{2m-2}}%
^{2m-2}\right)  }\right]  \text{.}%
\end{align}

\subsection{Reduced Multialgebras of a Resonant Submultialgebra}

In ref. \cite{sexpansion1} was shown that, if $S_{p}=\hat{S}_{p}\cup\check
{S}_{p}$ is a partition of the subsets $S_{p}\subset S$ that satisfy%
\begin{equation}
\check{S}_{p_{i}}\cap\hat{S}_{p_{i}}=\phi,\label{rrm1}%
\end{equation}
then%
\begin{equation}
\hat{S}_{p_{1}}\times\check{S}_{p_{2}}\times...\times\check{S}_{p_{n}}\subset%
{\displaystyle\bigcap\limits_{r\in i_{\left(  p_{1},...,p_{n}\right)  }}}
\hat{S}_{r}.\label{rrm2}%
\end{equation}
The conditions (\ref{rrm1}) and (\ref{rrm2}) induce the decomposition
$\mathfrak{G}_{R}=\mathfrak{\check{G}}_{R}\oplus\overset{\wedge}{\mathfrak{G}%
}_{R}$ on the resonant subalgebra, where
\begin{equation}
\mathfrak{\check{G}}_{R}=\oplus_{p\in I}\check{S}_{p}\otimes V_{p},\label{rrm3}%
\end{equation}%
\begin{equation}
\overset{\wedge}{\mathfrak{G}}_{R}=\oplus_{p\in I}\hat{S}_{p}\otimes
V_{p}.\label{rrm4}%
\end{equation}
When conditions (\ref{rrm1}) and (\ref{rrm2}) hold, then%
\begin{equation}
\left[  \overset{\wedge}{\mathfrak{G}}_{R},\mathfrak{\check{G}}_{R}%
,...,\mathfrak{\check{G}}_{R}\right]  _{S}\subset\overset{\wedge}%
{\mathfrak{G}}_{R},\label{rrm5}%
\end{equation}
and therefore $\left\vert \mathfrak{\check{G}}_{R}\right\vert $\ corresponds
to a reduced algebra of $\ \mathfrak{G}_{R}$.

The following theorem provides necessary conditions under which a reduced
multialgebra can be extracted from a resonant submultialgebra:

\begin{theorem}
\bigskip If $S_{p}=\hat{S}_{p}\cup\check{S}_{p}$ is a partition of the subsets
$S_{p}\subset S$ that satisfy%
\begin{equation}
\check{S}_{p_{i}}\cap\hat{S}_{p_{i}}=\phi,
\end{equation}%
\begin{equation}
\hat{S}_{p_{1}}\times\check{S}_{p_{2}}\times...\times\check{S}_{p_{n}}\subset%
{\displaystyle\bigcap\limits_{r\in i_{\left(  p_{1},...,p_{n}\right)  }}}
\hat{S}_{r},%
\end{equation}
then the generalized Maurer-Cartan equations for the Reduced Multialgebras of
a Resonant Submultialgebra are given by
\begin{equation}
\tilde{d}_{m}\omega^{\left(  c_{r},\check{\gamma}_{r}\right)  }=\frac
{1}{\left(  2m-2\right)  !}C_{\left(  a_{p_{1}}^{1},\check{\alpha}_{p_{1}}%
^{1}\right)  ...\left(  a_{p_{2m-2}}^{2m-2},\check{\alpha}_{p_{2m-2}}%
^{2m-2}\right)  }^{\ \ \ \ \ \ \ \ \ \ \ \ \ \left(  c_{r},\check{\gamma}%
_{r}\right)  }\omega^{\left(  a_{p_{1}}^{1},\check{\alpha}_{p_{1}}^{1}\right)
}...\omega^{\left(  a_{p_{2m-2}}^{2m-2},\check{\alpha}_{p_{2m-2}}%
^{2m-2}\right)  }.
\end{equation}
\end{theorem}
\begin{proof}

If $S_{p}=\hat{S}_{p}\cup\check{S}_{p}$ is a partition of the subsets
$S_{p}\subset S$ that satisfy%
\begin{align}
\check{S}_{p}\cap\hat{S}_{p} &  =\varnothing,\label{w7}\\
\check{S}_{p}\times\hat{S}_{q} &  =%
{\displaystyle\bigcap\limits_{r\in i\left\{  p,q\right\}  }}
\hat{S}_{r},
\end{align}
then
\begin{equation}
\omega^{a_{p}}=\sum_{\lambda_{\check{\alpha}_{p}}\in\check{S}_{p}}%
\lambda_{\check{\alpha}_{p}}\omega^{\left(  a_{p},\check{\alpha}_{p}\right)
}+\sum_{\lambda_{\hat{\alpha}_{p}}\in\hat{S}_{p}}\lambda_{\hat{\alpha}_{p}%
}\omega^{\left(  a_{p},\hat{\alpha}_{p}\right)  },\label{w8}%
\end{equation}
where the set of indices $\left\{  \alpha_{p}\right\}  =\left\{  \check
{\alpha}_{p},\hat{\alpha}_{p}\right\}  $ is such that $\lambda_{\check{\alpha
}_{p}}\in\check{S}_{p}$ and $\lambda_{\hat{\alpha}_{p}}\in\hat{S}_{p}$. This
means that the dual resonant submultialgebra $\mathfrak{G}_{R}^{\ast}$ is
generated by the forms%
\begin{equation}
\left\{  \omega^{\left(  a_{p},\alpha_{p}\right)  }\right\}  =\left\{
\omega^{\left(  a_{p},\check{\alpha}_{p}\right)  },\omega^{\left(  a_{p}%
,\hat{\alpha}_{p}\right)  }\right\}, \label{w9}%
\end{equation}
so that $\mathfrak{G}_{R}^{\ast}$ is given by
\begin{equation}
\mathfrak{G}_{R}^{\ast}=V_{0}^{\ast}\oplus V_{1}^{\ast},\label{w10}%
\end{equation}
where $V_{0}^{\ast}=\left\{  \omega^{\left(  a_{p},\check{\alpha}_{p}\right)
}\right\}  $, $V_{1}^{\ast}=\left\{  \omega^{\left(  a_{p},\hat{\alpha}%
_{p}\right)  }\right\}  .$ \ The reduction condition is given by the condition
$\left[  V_{0},V_{1}\right]  \subset V_{1}$ or equivalently $C_{\left(
b_{r},\check{\beta}_{r}\right)  \left(  c_{t},\hat{\gamma}_{t}\right)
}^{\left(  a_{s},\check{\alpha}_{s}\right)  }=0$: Since
\begin{equation}
C_{\left(  b_{r},\check{\beta}_{r}\right)  \left(  c_{t},\hat{\gamma}%
_{t}\right)  }^{\left(  a_{s},\check{\alpha}_{s}\right)  }=K_{\check{\beta
}_{r}\hat{\gamma}_{t}}^{\check{\alpha}_{s}}C_{b_{r}c_{t}}^{a_{s}},\label{w11}%
\end{equation}
and that (\ref{w7}) says to us that $\lambda_{\check{\beta}_{r}}\in\check
{S}_{r}$, $\lambda_{\hat{\gamma}_{t}}\in\hat{S}_{t}$ we have
\begin{equation}
\lambda_{\check{\beta}_{r}}\lambda_{\hat{\gamma}_{t}}=K_{\check{\beta}_{r}%
\hat{\gamma}_{t}}^{\alpha_{r}}\lambda_{\alpha_{r}}\in%
{\displaystyle\bigcap\limits_{r\in i\left\{  p,q\right\}  }}
\hat{S}_{r},\label{ww12}%
\end{equation}
which imply $K_{\check{\beta}_{r}\hat{\gamma}_{t}}^{\check{\alpha}_{s}}=0$ and
therefore
\begin{equation}
C_{\left(  b_{r},\check{\beta}_{r}\right)  \left(  c_{t},\hat{\gamma}%
_{t}\right)  }^{\left(  a_{s},\check{\alpha}_{s}\right)  }=0\text{.}%
\end{equation}
This means that the set $\left\{  \omega^{\left(  a_{p},\check{\alpha}%
_{p}\right)  }\right\}  $ generates the so-called dual reduced multialgebra of
a resonant submultialgebra. The corresponding Maurer-Cartan equations for this
reduced multialgebra are%
\begin{equation}
\tilde{d}_{m}\omega^{\left(  c_{r},\check{\gamma}_{r}\right)  }=\frac
{1}{\left(  2m-2\right)  !}C_{\left(  a_{p_{1}}^{1},\check{\alpha}_{p_{1}}%
^{1}\right)  ...\left(  a_{p_{2m-2}}^{2m-2},\check{\alpha}_{p_{2m-2}}%
^{2m-2}\right)  }%
^{\ \ \ \ \ \ \ \ \ \ \ \ \ \ \ \ \ \ \ \ \ \ \ \ \ \ \ \left(
c_{r},\check{\gamma}_{r}\right)  }\omega^{\left(  a_{p_{1}}^{1},\check{\alpha
}_{p_{1}}^{1}\right)  }\cdot\cdot\cdot\omega^{\left(  a_{p_{2m-2}}%
^{2m-2},\check{\alpha}_{p_{2m-2}}^{2m-2}\right)  }.
\end{equation}
\end{proof}
In the compact notation we have
\begin{equation}
\tilde{d}_{m}\theta^{\left(  S\right)  }=\frac{1}{\left(  2m-2\right)
!}\left[  \theta^{\left(  S\right)  },\theta^{\left(  S\right)  }%
,\overset{2m-2}{\cdot\cdot\cdot\cdot},\theta^{\left(  S\right)  }\right]
\text{,}%
\end{equation}
where
\begin{equation}
\mathbf{\ }\theta^{\left(  S\right)  }=\omega^{\left(  c_{r},\check{\gamma
}_{r}\right)  }X_{\left(  c_{r},\check{\gamma}_{r}\right)  },%
\end{equation}
\begin{equation}
\left[  \theta^{\left(  S\right)  },\theta^{\left(  S\right)  },\overset
{2m-2}{\cdot\cdot\cdot\cdot},\theta^{\left(  S\right)  }\right]
=\omega^{\left(  a_{p_{1}}^{1},\check{\alpha}_{p_{1}}^{1}\right)  }\wedge
\cdot\cdot\cdot\wedge\omega^{\left(  a_{p_{2m-2}}^{2m-2},\check{\alpha
}_{p_{2m-2}}^{2m-2}\right)  }\left[  X_{\left(  a_{p_{1}}^{1},\check{\alpha
}_{p_{1}}^{1}\right)  },\cdot\cdot\cdot,X_{\left(  a_{p_{2m-2}}^{2m-2}%
,\check{\alpha}_{p_{2m-2}}^{2m-2}\right)  }\right]  \text{.}%
\end{equation}

\subsection{Recovering results of section 3}

Now we comment that the expansion method developed in section 3 can be
recovered in the S-expansion formalism for a particular election of the
semigroup. For example, we will show that the equations (\ref{g8})%
\begin{equation}
\tilde{d}_{m}\omega^{k_{s},\alpha}=\frac{1}{\left(  2m-2\right)  !}C_{\left(
i_{p_{1}}^{1},\beta^{1}\right)  \cdot\cdot\cdot\left(  i_{p_{2m-2}}%
^{2m-2},\beta^{2m-2}\right)  }^{\left(  k_{s},\alpha\right)  }\omega
^{i_{p_{1}}^{1},\beta^{1}}\wedge\cdot\cdot\cdot\wedge\omega^{i_{p_{2m-2}%
}^{2m-2},\beta^{2m-2}},\label{eqiv1} 
\end{equation}
\begin{equation}
C_{\left(  i_{p_{1}}^{1},\beta^{1}\right)  \cdot\cdot\cdot\cdot\cdot\left(
i_{p_{2m-2}}^{2m-2},\beta^{2m-2}\right)  }^{\left(  k_{s},\alpha\right)
}=C_{i_{p_{1}}^{1}\cdot\cdot\cdot\cdot i_{p_{2m-2}}^{2m-2}}^{k_{s}}%
\delta_{\beta^{1}+\cdot\cdot\cdot\cdot+\beta^{2m-2}}^{\alpha}\text{,}
\end{equation}
can be recovered in the language of S-expansions.

In fact, let us choose the following semigroup:%
\begin{equation}
S_{E}^{\left(  N\right)  }=\left\{  \lambda_{\alpha}\text{, }\alpha
=0,...,N+1\right\}, \label{r3}%
\end{equation}
with a multiplication rule given by%
\begin{equation}
\lambda_{\alpha}\lambda_{\beta}=\lambda_{H_{N+1}\left(  \alpha+\beta\right)
}=\delta_{H_{N+1}\left(  \alpha+\beta\right)  }^{\gamma}\lambda_{\gamma
}.\label{r4}%
\end{equation}

The two-selectors for $S_{E}^{\left(  N\right)  }$ read
\begin{equation}
K_{\alpha\beta}^{\gamma}=\delta_{H_{N+1}\left(  \alpha+\beta\right)  }%
^{\gamma},\label{r5}%
\end{equation}
where $\delta_{\sigma}^{\rho}$ is the Kronecker delta. The multiplication rule
(\ref{r4}) can be directly generalized to%
\begin{align}
\lambda_{\alpha_{1}}....\lambda_{\alpha_{n}}  & =\lambda_{H_{N+1}\left(
\alpha_{1}+...+\alpha_{n}\right)  }=\delta_{H_{N+1}\left(  \alpha
_{1}+...+\alpha_{n}\right)  }^{\gamma}\lambda_{\gamma}, \\
K_{\alpha_{1}...\alpha_{n}}^{\ \ \ \ \ \ \ \ \ \gamma}  & =\delta
_{H_{N+1}\left(  \alpha_{1}+...+\alpha_{n}\right)  }^{\gamma}.
\end{align}

Consider now a higher order Lie algebra $\left(  \mathcal{G},\left[
,...,\right]  \right)  $ of order $n=2m-2$, whose generalized MC equations are
given by
\begin{equation}
\tilde{d}_{m}\omega^{A}=\frac{1}{\left(  2m-2\right)  !}C_{B_{1}\cdot
\cdot\cdot\cdot\cdot\cdot B_{2m-2}}%
^{\ \ \ \ \ \ \ \ \ \ \ \ \ \ A}\omega^{B_{1}}\cdot\cdot
\cdot\cdot\cdot\omega^{B_{2m-2}}.
\end{equation}
Then the generalized MC equations of the $S_{E}^{\left(  N\right)  }$-expanded
Lie multialgebra are given by%
\begin{equation}
\tilde{d}_{m}\omega^{\left(  A,\alpha\right)  }=\frac{1}{\left(  2m-2\right)
!}C_{\left(  B_{1},\beta_{1}\right)  ...\left(  B_{2m-2},\beta_{2m-2}\right)
}^{\ \ \ \ \ \ \ \ \ \ \ \ \ \left(  A,\alpha\right)  }\omega^{\left(
B_{1},\beta_{1}\right)  }...\omega^{\left(  B_{2m-2},\beta_{2m-2}\right)  },%
\end{equation}
\begin{equation}
C_{\left(  B_{1},\beta_{1}\right)  ......\left(  B_{2m-2},\beta_{2m-2}\right)
}^{\ \ \ \ \ \ \ \ \ \ \ \ \ \left(  A,\alpha\right)  }=C_{B_{1}%
......B_{2m-2}}^{\ \ \ \ \ \ \ \ \ \ \ \ \ A}\delta_{H_{N+1}\left(  \beta
_{1}+...+\beta_{2m-2}\right)  }^{\gamma},%
\end{equation}
where $\alpha,\beta_{1},...,\beta_{2m-2}=0,1,...,N,N+1$.

In section 4.3, we used latin indices $i,j,k$ when we restrict the greek
indices (of the semigroup elements) to take values in $\left\{
0,1,...,N\right\}  $, following the convention adopted in \cite{sexpansion},
\cite{dualSexpansion} and \cite{sexpansion1}. However, in section 3 latin
indices were used to label the basis elements of the algebra or multialgebra
and their dual forms. This was done so in order to obtain a direct
generalization of the expansion method \cite{deazcarraga1} to the higher order
Lie algebra case. To make a consistent comparison we will not use latin
indices to perform the $0$-reduction. Instead we continue to use greek
indices, but write explicitly that they are restricted to take values in
$\left\{  0,1,...,N\right\}  $.

Therefore, when the greek indices cannot take the value $N+1$, we have
\begin{equation}
\delta_{H_{N+1}\left(  \beta_{1}+...+\beta_{2m-2}\right)  }^{\gamma}%
=\delta_{\beta_{1}+...+\beta_{2m-2}}^{\gamma},%
\end{equation}
and the $0$-reduced multialgebra of the $S_{E}^{\left(  N\right)  }$-expanded
Lie multialgebra is given by the following generalized Maurer-Cartan equations%
\begin{equation}
\tilde{d}_{m}\omega^{\left(  A,\alpha\right)  }=\frac{1}{\left(  2m-2\right)
!}C_{\left(  B_{1},\beta_{1}\right)  ...\left(  B_{2m-2},\beta_{2m-2}\right)
}^{\ \ \ \ \ \ \ \ \ \ \ \ \ \left(  A,\alpha\right)  }\omega^{\left(
B_{1},\beta_{1}\right)  }...\omega^{\left(  B_{2m-2},\beta_{2m-2}\right)
},\label{equiv2}%
\end{equation}%
\begin{align}
C_{\left(  B_{1},\beta_{1}\right)  ......\left(  B_{2m-2},\beta_{2m-2}\right)
}^{\ \ \ \ \ \ \ \ \ \ \ \ \ \left(  A,\alpha\right)  }  & =C_{B_{1}%
......B_{2m-2}}^{\ \ \ \ \ \ \ \ \ \ \ \ \ A}\delta_{\beta_{1}+...+\beta
_{2m-2}}^{\gamma},\\
\alpha,\beta_{1},...,\beta_{2m-2}  & =0,1,...,N\text{.} \nonumber%
\end{align}
The equivalence between (\ref{eqiv1}) and (\ref{equiv2}) is explicit if we
consider that the vector space of the original multialgebra is split into a
sum of two vector spaces $\mathcal{G=}V_{0}\oplus V_{1}.$ \ Then the dual
basis is descomposed as $\left\{  \omega^{A}\right\}  =$ $\left\{
\omega^{i_{0}}\right\}  \cup\left\{  \omega^{i_{1}}\right\}  $ where
$\omega^{i_{0}}\in V_{0}$ and $\omega^{i_{1}}\in V_{1}$.

\section{Comments and Possible Developments}

We have shown that the successful expansion methods developed in refs.
\cite{deazcarraga1}, \cite{dualSexpansion} (see also \cite{hatsuda},
\cite{deazcarraga2}) can be generalized so that they permit obtaining new
higher-order Lie algebras of increasing dimensions from $\left(
\mathcal{G},\left[  ,...,\right]  \right)  $ by a geometric procedure based on
expanding the generalized Maurer-Cartan equations.

The main results of this paper are: the generalizations of the expansion
methods developed in refs. \cite{deazcarraga1}, \cite{dualSexpansion} and we
give the general structure of the expansion method, as well as to construct
the dual S-expansion procedure of higher-order Lie algebras.

The expansion procedures considered here could play an important role in the
context of gravity in higher dimensions. \ In fact, it seems likely that it is
possible, in the context of a Chern-Simons action, to construct a theory that
describes a consistent coupling of higher-spin fields to a particular form of
Lovelock gravity.

This work was supported in part by Direcci\'{o}n de Investigaci\'{o}n,
Universidad de Concepci\'{o}n through Grant \# 208.011.048-1.0 and in part by
FONDECYT through Grants \#s 1080530 and 1070306 . Two of the authors (R.C. and
N.M) were supported by grants from the Comisi\'{o}n Nacional de
Investigaci\'{o}n Cient\'{\i}fica y Tecnol\'{o}gica CONICYT and from the
Universidad de Concepci\'{o}n, Chile. One of the authors (A.P) wishes to thank
to S. Theisen for his kind hospitality at the M.P.I f\"{u}r Gravitationsphysik
in Golm where part of this work was done. He is also grateful to German
Academic Exchange Service (DAAD) and Consejo Nacional de Ciencia y
Tecnolog\'{\i}a (CONICYT) for financial support.

\end{document}